\documentclass[11pt,english]{amsart}
\usepackage[T1]{fontenc}
\usepackage[latin9]{inputenc}
\usepackage{amsthm}
\usepackage{amstext}
\usepackage{stmaryrd}
\usepackage{esint}
\usepackage{mathrsfs}
\makeatletter
\numberwithin{equation}{section}
\numberwithin{figure}{section}
\theoremstyle{plain}
\newtheorem{thm}{\protect\theoremname}[section]
  \theoremstyle{plain}
  \newtheorem{lem}[thm]{\protect\lemmaname}

\usepackage{amsmath,amsfonts,amssymb,amsthm,epsfig}

\usepackage{color}

\usepackage[final,allcolors=blue,colorlinks=true]{hyperref}
\usepackage[leqno]{amsmath}

\newcommand{\ds}{\displaystyle}

\def\R{\mathbb R}

\def\C{\mathbb C}

\numberwithin{equation}{section}
\theoremstyle{definition}
\begingroup
\newtheorem{rem}[thm]{Remark}
\newtheorem{prop}[thm]{Proposition}

\endgroup

\makeatother

\usepackage{babel}
  \providecommand{\lemmaname}{Lemma}
\providecommand{\theoremname}{Theorem}

\begin{document}

\title{\title[ linearly coupled Schr\"{o}dinger systems]
~Infinitely many solutions to linearly coupled Schr\"{o}dinger systems with non-symmetric potentials}

\author{Chun-Hua Wang, Jing Yang}

\date{\today}

\address{[Chunhua Wang] School of Mathematics and Statistics, Central China
Normal University, Wuhan 430079, P. R. China. }

\email{[Chunhua Wang] chunhuawang@mail.ccnu.edu.cn}

\address{[Jing Yang] School of Mathematics and Statistics, Central China
Normal University, Wuhan 430079, P. R. China.}

\email{[Jing Yang] yyangecho@163.com}
\begin{abstract}
We study a linearly coupled Schr\"{o}dinger system in $\R^N(N\leq3).$
Assume that the potentials in the system are continuous functions satisfying suitable decay assumptions,
but without any symmetry properties and the parameters in the system satisfy some suitable restrictions.
Using the Liapunov-Schmidt reduction methods two times and combing localized energy method, we prove that the  problem has infinitely many positive synchronized
solutions, which extends the result Theorem 1.2 about nonlinearly coupled Schr\"{o}dinger equations
in \cite{aw} to our linearly coupled problem.
\end{abstract}

\maketitle
{\small
\noindent {\bf Keywords:} linearly coupled; non-symmetric potentials; synchronized
solutions.

\smallskip

\tableofcontents{}

\section{ Introduction and main result }\label{s1}

In this paper, we consider the following
nonlinear Schr\"{o}dinger system in $\R^N(N\leq3),$
\begin{equation}\label{1.1}
\left\{%
\begin{array}{ll}
    -\Delta u+(1+\epsilon P(x))u=u^3+\beta v,   & \hbox{$x\in \R^N$}, \vspace{0.2cm}\\
    -\Delta v+(1+\epsilon Q(x))v=v^3+\beta u, & \hbox{$x\in\R^N$},
\end{array}%
\right.\end{equation}
where the potentials $P(x),Q(x)$ are continuous functions satisfying suitable decay assumptions,
but without any symmetry properties, $\epsilon$ is a positive constant, $\beta\in \R$ is a coupling constant. We are mainly interested in the existence of infinitely
many positive synchronized solutions of system \eqref{1.1}.

Systems of nonlinear Schr\"{o}dinger equations have been received a great deal of attention and significant progress has been made in recent years. The results one can achieve depend on
the way that the system is coupled. The case in which the coupling
is nonlinear has been studied extensively, which is motivated
by applications to nonlinear optic and Bose-Einstein condensation. See for example
\cite{ac,ac2,aw,bd,bw,msg,dw,lw,lw2,lw3,mmp,ntt,pw,s,wt,ww} and references therein.
Recently in \cite{pw}, Peng and Wang studied the following
nonlinearly coupled Schr\"{o}dinger equations
\begin{equation}\label{pq2}
\left\{%
\begin{array}{ll}
    -\Delta u+P(x)u=\mu_{1}|u|^2u+\beta v^{2}u,   & \hbox{$x\in \R^3$}, \vspace{0.2cm}\\
    -\Delta v+Q(x)v=\mu_{2}|v|^2v+\beta u^{2}v,& \hbox{$x\in\R^3$}.
\end{array}%
\right.
\end{equation}
Applying the finite reduction method, they obtained the existence of infinitely many solutions of segregated or synchronized type for radial symmetric potentials $P(|x|)$, $Q(|x|)$ satisfying some algebra decay assumptions.
For $N=2$, the first auhtor ect in \cite{wxzz}
constructed an unbounded sequence of non-radial positive vector solutions
 of segregated type when $\beta$ is in some suitable interval, which gives a positive answer
 to an interesting problem raised by Peng and Wang in Remark 4.1 in \cite{pw}.
In \cite{aw}, Ao and Wei obtained the existence of infinitely many solutions \eqref{pq2}
for nonsymmetric potentials $P(x)$, $Q(x)$ satisfying some exponential decay assumptions.

 In this paper, we are mainly interested in a class of nonlinear
 Schr\"{o}dinger equations which are linearly coupled. Systems of
 this type arise in nonlinear optics.
 For example, the propagation of optical pulses in nonlinear
 dual-core fiber can be described by two linearly
 coupled nonlinear Schr\"{o}dinger equations like
 \begin{equation}\label{1.2}
\left\{%
\begin{array}{ll}
-i\frac{\partial}{\partial t}\Phi_1=\Delta\Phi_1-V_1(x)\Phi_1+|\Phi_1|^2\Phi_1
+\beta\Phi_2,~~&\text{in}~\R^N, \vspace{0.2cm}\\
-i\frac{\partial}{\partial t}\Phi_2=\Delta\Phi_2-V_2(x)\Phi_1+|\Phi_1|^2\Phi_2
+\beta\Phi_1,~~&\text{in}~\R^N,\vspace{0.2cm}\\
    \Phi_j=\Phi_j(x,t)\in \C, t>0,&j=1,2,
\end{array}%
\right.
\end{equation}
where $N\leq 3,\Phi_1(t,x)$ and $\Phi_2(t,x)$ are the complex valued envelope functions, and $\beta$
is the coupling coefficient between the two cores, whose sign
determines whether the interactions of fiber couplers are repulsive or attractive. In the attractive case the components of a vector solution tend to go along with each other leading
to synchronization, and in the repulsive case the components tend to segregate with each
other leading to phase separations. These phenomena have been documented in numeric
simulations (e.g.,\cite{aa} and references therein).

We will look for standing waves of the form
 \begin{equation}\label{1.3}
\Phi_1(t,x)=u(x)e^{it},\,\,\,\,\Phi_2(t,x)=v(x)e^{it},
\end{equation}
where $u(x)$ and $v(x)$ are real valued functions.
Substituting \eqref{1.3}
into \eqref{1.2} and setting $V_{1}(x)=\epsilon P(x),V_{2}(x)=\epsilon Q(x),$ we are led to \eqref{1.1}.

For linearly coupled nonlinear Schr\"{o}dinger equations,
to our knowledge it seems that there are very few results. One can refer to \cite{aa,a,acr1,acr2,b,lp}.
When the dimension of the space $N=1$, for $\epsilon=0,\beta<0,$
 \eqref{1.1} has in addition to the semi-trivial solutions $(\pm U,0),(0,\pm U),$ two types of solitons like solutions given by
 \begin{align*}
&(U_{1+\beta},U_{1+\beta}),\,\,\,(-U_{1+\beta},-U_{1+\beta})\,\,\,\text{for}\,\,\,-1\leq \beta\leq 0,\,\,\,\text{symmetric~states},\\
 &(U_{1-\beta},U_{1-\beta}),\,\,\,(-U_{1-\beta},-U_{1-\beta})\,\,\,\text{for}\,\,\,\beta\leq 0,\,\,\,\text{anti-symmetric~states},
\end{align*}
 where for $\lambda>0,$ $U_{\lambda}$ is the unique solution of
 $$
 \left\{%
\begin{array}{ll}
- u''+\lambda u=u^{3}, \,\,\,\,u>0\,\,\,\text{in}\,\, \R,\vspace{0.2cm}\\
    u(0)=\max_{x\in \R}u(x),\,\,\,\,\,u(x)\in H^{1}(\R).
\end{array}%
\right.
$$
In \cite{aa},
 Akhmediev and Ankiewicz observed numerically that for
 $-1<\beta<0,$ there exists a family of new solutions for \eqref{1.1} with $\epsilon=0,$ bifurcating from the branch
 of the anti-symmetric state $\beta=-1.$ This kind of results were rigorously verified in \cite{ac}
 for small value of the parameter $\beta=0.$ More precisely, Ambrosetti and Colorado in \cite{ac}
 proved that a solution with one 2-bump component having bumps located near $-|\ln(-\beta) |$ and $|\ln(-\beta) |$ while
 the other component having one negative peak exists.
 In \cite{b}, Buffoni gave a rigorous proof of the existence of the bifurcation described by Ambrosetti,  Arcoya and G\'{a}mez in \cite{aag}

In \cite{acr2}, Ambrosetti, Cerami and Ruiz
 studied \eqref{1.1} in the case $\epsilon=0.$
They proved that if $\mathcal {P}$ denotes a regular
polytope centered at the origin of $\R^{N}$ such
that its side is greater than radius, then there exists
a solution with one multi-bump
component having bumps located near the vertices of $\xi \mathcal {P},$
where $\xi\sim \log(\frac{1}{\beta})$, while the other component has one negative peak.
Taking $\epsilon=0$ and substituting $u^{3},v^{3}$ by $(1+a(x))|u|^{p-1}u,(1+b(x))|v|^{p-1}v(N\geq 2,1<p<2^{*}-1)$ respectively, in \cite{acr1}, Ambrosetti, Colorado and Ruiz obtained some results about the
existence of positive ground and bound state of \eqref{1.1}.
Recently in \cite{cz}, Chen and Zou studied the following linearly coupled Schr\"{o}dinger  equations
\begin{equation}\label{cz}
\left\{%
\begin{array}{ll}
-\epsilon^{2}\Delta u+a(x)u=u^{p}+\lambda v, \,\,\,\,&x\in \R^{N},\vspace{0.2cm}\\
-\epsilon^{2}\Delta v+b(x)v=v^{2^{*}-1}+\lambda u, \,\,\,\,&x\in \R^{N},\vspace{0.2cm}\\
    u,v>0,\,\,\,u(x),v(x)\rightarrow 0,\,\,\,\,&\text{as}\,\,\,|x|\rightarrow\infty,
\end{array}%
\right.
\end{equation}
where $N\geq 3$ and $a(x),b(x)$ are positive potentials which are both bounded away from 0.
Under some conditions of $a(x),b(x)$ and $\lambda>0,$ they obtained positive solutions for \eqref{cz}
when $\epsilon>0$ small enough, which has concentration phenomenon as $\epsilon\rightarrow 0.$
Very recently, in \cite{lp} Lin and Peng studied linearly coupled nonlinear Schr\"{o}dinger systems similar to
\eqref{1.1} with $k>=2$
equations and $\epsilon=0$. They examined the effect of the linear coupling to the solution structure. When $N=2,3$, for any prescribed integer,
they constructed a non-radial vector solutions of segregated type, with each component having exactly $l$ positive bumps for $\beta>0$
sufficiently small. They also gave an explicit description on the characteristic features of the vector solutions.

Inspired by \cite{aw,pw}, we want to investigate the existence of infinitely
many positive synchronized solutions of system \eqref{1.1}. In order to state our main result, now
we give the conditions imposed on $P(x),Q(x)$ which are similar to those in \cite{aw},\\
{\bf ($K_1$)}\,\,  $\ds\lim_{|x|\rightarrow\infty}P(x)=\ds\lim_{|x|\rightarrow\infty}Q(x)=0;$ \\
{\bf ($K_2$)}\,\, $\exists~0<\alpha<1,$ $\ds\lim_{|x|\rightarrow\infty}(\gamma^2P(x)+\gamma^2Q(x))e^{\alpha\gamma|x|}=+\infty$,\\
where $\gamma$ is defined in \eqref{1.7} below.

The energy functional associated with problem  \eqref{1.1} is
\begin{equation}\label{1.5}
\arraycolsep=1.5pt
\begin{array}{rl}\displaystyle
J(u,v)&=\ds\frac{1}{2}\ds\int_{\R^N}\big[|\nabla u|^2+(1+\epsilon P(x))u^2+|\nabla v|^2+(1+\epsilon Q(x))v^2\big]-\frac{1}{4}\ds\int_{\R^N}(u^4+v^4)\\[5mm]
&\,\,\,\,\,\,\,\,-\beta\ds\int_{\R^N}uv,\,
u,v\in H^1(\R^N).
\end{array}
\end{equation}
We will study $J(u,v)$ in Section \ref{s6}. Let $w$ be the unique solution of
\begin{equation}\label{1.6}\left\{%
\begin{array}{ll}
\Delta w-w+w^3=0,\,\,w>0,\,&\text{in}~\R^N, \vspace{0.2cm}\\
w(0)=\ds\max_{x\in\R^N}(x),\,\,w\rightarrow0,\,&\text{as}\,|x|\rightarrow\infty.
\end{array}%
\right.
\end{equation}
By the well-known result of Gidas, Ni and Nirenberg in \cite{gs}, $w$ is radially symmetric and strictly decreasing,
$w'(r)<0$ for $r>0$. Moreover, from \cite{gs} we know the
following asymptotic behavior of $w$:
$$\left\{%
\begin{array}{ll}
w(r)=A_Nr^{-\frac{N-1}{2}}e^{-r}(1+O(\frac{1}{r})), \vspace{0.2cm}\\
w'(r)=-A_Nr^{-\frac{N-1}{2}}e^{-r}(1+O(\frac{1}{r})),
\end{array}%
\right.$$
for $r>0$, where $A_N$ is a positive constant.

Noting that if $\beta<1$ and
 \begin{equation}\label{1.7}
\gamma=\sqrt{1-\beta},
\end{equation}
then it follows from \cite{wy} that
\begin{equation}\label{1.8}
(U,V)=(\gamma w(\gamma x),\gamma w(\gamma x))
\end{equation}
solves the following problem
\begin{equation}\label{1.9}\left\{%
\begin{array}{ll}
-\Delta u+(1-\beta)u=u^3,\,\,\,\,\,\text{in}~\R^N, \vspace{0.2cm}\\
-\Delta v+(1-\beta)v=v^3,\,\,\,\,\,\,\text{in}~\R^N.
\end{array}%
\right.
\end{equation}

 We will use $(U,V)$ as the building blocks for the solutions of \eqref{1.1}.
 Let $\mu>0$ be a real number such that $w(x)\leq ce^{-|x|}$ for $|x|>\mu$ and
some constant $c$ independent of $\mu$ large. Denote
$\mathbf{P}_m=(P_1,\cdots,P_m).$ Now we define the
configuration space
$$
\Omega_1=\R^N,\,\,\,\Omega_m=\Big\{\mathbf{P}_m\in\R^{mN}\big|\,\,\ds\min_{j\neq k}|P_j-P_k|\geq\frac{\mu}{\gamma}\Big\},\,\forall m>1.
$$

For $\mathbf{P}_m\in\Omega_m$, we define
 \begin{equation}\label{1.10}
(U_{P_j},V_{P_j})=(U(x-P_j),V(x-P_j))
\end{equation}
and the approximate solutions to be
 \begin{equation}\label{1.11}
U_{\mathbf{P}_m}=\ds\sum_{j=1}^{m}U_{P_j},\,\,\,V_{\mathbf{P}_m}=\ds\sum_{j=1}^{m}V_{P_j}.
\end{equation}
Denote
\begin{equation}\label{1.12}
G\Big(
\begin{array}{ccccc}
u \\[1mm]
v\\
\end{array}
\Big):=\Big(
\begin{array}{ccccc}
\Delta u-(1+\epsilon P(x))u+u^3+\beta v \\[1mm]
\Delta v-(1+\epsilon Q(x))v+v^3+\beta u \\
\end{array}
\Big)
\end{equation}
and for $f=\Big(
\begin{array}{ccccc}
f_1 \\
f_2\\
\end{array}
\Big)$, $g=\Big(
\begin{array}{ccccc}
g_1 \\
g_2\\
\end{array}
\Big)$, we denote $
 \left \langle f,g\right\rangle=\ds\int_{\R^N}(f_1g_1+f_2g_2).
$

Now we state our main result as follows:
 \begin{thm}\label{thm1.1}
Let $(K_1)$ and $(K_2)$ hold. Then there exist $\epsilon_0$ and $\beta^*>0$ such that for
 $\beta\in(-\beta^*,0)\cup (0,1)$, and $0<\epsilon<\epsilon_0$,
problem~\eqref{1.1} has infinitely many positive
synchronized solutions.
\end{thm}

\begin{rem}
We only consider synchronized solutions for \eqref{1.1} in the nonsymmetric case.
However, we do not know whether we can also obtain infinitely many segregated solutions for \eqref{1.1}
as \cite{pw} by our method. It would be very interesting to consider this problem.
\end{rem}

Our result provides a new phenomenon for \eqref{1.1} in the nonsymmetric case with linearly coupled terms.
 Our method is different from \cite{pw}.
To the best knowledge of us, our result is new.

In order to prove Theorem \ref{thm1.1}, we mainly use the Liapunov-Schmidt reduction method as in \cite{aw,awz,lnw}. There are two main difficulties.
Firstly, we need to show that the maximum points will not go to infinity (see Section \ref{s6}). This is guaranteed by the slow decay assumption $(K_{2}).$ Secondly, we have to detect the difference in the energy when the spikes move to the boundary of the configuration space.
A crucial estimate is Lemma \ref{lem4.2}, in which we prove that the accumulated error can be controlled from step $m$ to step $(m+1)$.
Compared with \cite{aw}, due to the linear coupling terms, there are new difficulties in estimates.

Our paper is organized as follows.
In section~\ref{s3}, we perform the first finite reduction. In Section \ref{s5}, we show a key estimate which majors the differences between the $m$-th step and the $(m+1)$-th step which involves a secondary Liapunov-Schmidt reduction. We prove Theorem \ref{thm1.1} in Section \ref{s6}.
Throughout this paper, denote $\mathbf{H}=H^1(\R^N)\times H^1(\R^N)$ and $\|(u,v)\|_{H^1(\R^N)}=\|u\|_{H^1(\R^N)}+\|v\|_{H^1(\R^N)}.$
$c$, $C$ will always denote various generic constants that are independent of $\mu$ for $\mu$ large.

\section{The first Liapunov-Schmidt reduction}\label{s3}

In this section, we perform a finite-dimensional reduction.

For $\mathbf{P}_m\in\Omega_m$, we define the following functions:

\begin{equation}\label{2.1}
D_{jk}=\Big(
\begin{array}{ccccc}
D_{jk,1} \\[1mm]
D_{jk,2}\\
\end{array}
\Big)=\left(
\begin{array}{ccccc}
\frac{\partial U_{P_j}}{\partial x_k}\zeta_j(x) \\[1mm]
\frac{\partial V_{P_j}}{\partial x_k}\zeta_j(x) \\
\end{array}
\right),\,\,\,\text{for}\,\,j=1,\cdots,m,\,k=1,\cdots,N,
\end{equation}
where $\zeta_j(x)=\zeta\big(\frac{2|x-P_j|}{\mu-1}\big)$ and $\zeta(t)$
is a cut-off function such that $\zeta(t)=1$ for $|t|\leq\frac{1}{\gamma}$ and
$\zeta(t)=0$ for $|t|\geq\frac{\mu^2}{\gamma(\mu^2-1)}$. So we know that the support of $D_{jk}$
belongs to $B_{\frac{\mu^2}{2\gamma(\mu+1)}}(P_j)$.

Consider the following linear problem: given  $h=\Big(
\begin{array}{ccccc}
h_1 \\
h_2\\
\end{array}
\Big)$, we find a function  $\Big(
\begin{array}{ccccc}
\phi \\
\psi\\
\end{array}
\Big)$ satisfying
\begin{equation}\label{2.2e}
\left\{%
\begin{array}{ll}
L\Big(
\begin{array}{ccccc}
\phi \\[1mm]
\psi\\
\end{array}
\Big):=\Big(
\begin{array}{ccccc}
\Delta \phi-(1+\epsilon P(x))\phi+3U_{\mathbf{P}_m}^2\phi+\beta \psi \\[1mm]
\Delta \psi-(1+\epsilon Q(x))\psi+3V_{\mathbf{P}_m}^2\psi+\beta \phi \\
\end{array}
\Big)=h+\ds\sum_{j=1}^{m}\sum_{k=1}^{N}c_{jk}D_{jk},\vspace{0.2cm}\\
\Big \langle  \Big(
\begin{array}{ccccc}
\phi \\[1mm]
\psi\\
\end{array}
\Big),\Big(
\begin{array}{ccccc}
D_{jk,1} \\[1mm]
D_{jk,2}\\
\end{array}
\Big)\Big\rangle=0 \,\,\,\text{for}\,\,j=1,\cdots,m,\,k=1,\cdots,N.
\end{array} %
\right.
\end{equation}
Letting $0<\nu<1$ and
$$
F:=\ds\sum_{\mathbf{P}_{m}\in \Omega_m}e^{-\nu\gamma|\cdot-P_j|},
$$
we define the norm
\begin{equation}\label{2.3}
\|h\|_*=\sup_{x\in\R^N}\big|F(x)^{-1}h_1(x)\big|+\sup_{x\in\R^N}\big|F(x)^{-1}h_2(x)\big|.
\end{equation}

From \cite{dwy}, there holds the following relation
  \begin{equation}\label{linfinity}
\|u\|_{L^{\infty}(\R^{N})}\leq C\|u\|_{*},
\end{equation}
where $C>0$ independent of $\mu,m$ and $\textbf{P}_{m}.$

Firstly applying (iii) of Lemma 3.4 in \cite{acr1}, we give the following non-degeneracy result
which will be used later.
\begin{lem} \label{lem4.1}
There exists $\beta^*>0$ such that for $\beta\in(-\beta^*,0)\cup(0,1)$, $(U,V)$ is non-degenerate for the
system \eqref{1.9} in $\mathbf{H}$ in the sense that the kernel is given by
$$\text{Span}\Big\{\Big(\frac{\partial U}{\partial x_j},\frac{\partial V}{\partial x_j}\Big)\,\Big|\,j=1,\cdots,N\Big\}.$$
\end{lem}

In the following, $\sigma$ will denote a positive constant depending on $\epsilon$ and $\nu$
but independent of $\mu,m,\textbf{P}_{m}$ and may vary from line to line.

\begin{prop}\label{prop2.1}
Let $h$ with $\|h\|_*$  norm bounded and assume that $\Big(\Big(
\begin{array}{ccccc}
\phi \\
\psi\\
\end{array}
\Big),\{c_{jk}\}\Big)$ is a solution to problem \eqref{2.2e}. Then there exist positive numbers $\epsilon_0$, $\mu_0$,
$c\geq0$ such that for all $0<\epsilon<\epsilon_0$, $\mu\geq \mu_0$ and $\mathbf{P}_m\in \Omega_m$, we have
$$
\|(\phi,\psi)\|_*\leq c\|h\|_*,
$$
where $C$ is a positive constant independent of $\mu,m$ and $\textbf{P}_{m}\in \Omega_{m}.$
\end{prop}

\begin{proof}
Using the same argument as \cite{awz}, we prove this proposition by contradiction. Assume that there exists a solution $\Big(
\begin{array}{ccccc}
\phi \\[1mm]
\psi\\
\end{array}
\Big)$,  such that
$\|h\|_*\rightarrow0,\,\,\,\text{and}\,\,\|(\phi,\psi)\|_*=1.$

Multiplying the frist system \eqref{2.2e} by $D_{jk}=\Big(
\begin{array}{ccccc}
D_{jk,1} \\
D_{jk,2}\\
\end{array}
\Big)$ and integrating in $\R^N$, we get
$$
\ds\int_{\R^N}L\Big(
\begin{array}{ccccc}
\phi \\[1mm]
\psi\\
\end{array}
\Big) \Big(
\begin{array}{ccccc}
D_{jk,1} \\[1mm]
D_{jk,2}\\
\end{array}
\Big)=\ds\int_{\R^N}\Big(
\begin{array}{ccccc}
h_1 \\[1mm]
h_2\\
\end{array}
\Big) \Big(
\begin{array}{ccccc}
D_{jk,1} \\[1mm]
D_{jk,2}\\
\end{array}
\Big)+c_{jk}\ds\int_{\R^N}D_{jk}^2.
$$

By the definition of $D_{jk}$, letting $y=\gamma(x-P_j)$, we deduce
\begin{align*}
&\ds\int_{\R^N}D_{jk}^2=\ds\int_{\R^N}\left(D_{jk,1}^2+D_{jk,2}^2\right)
=2\gamma^{4-N}\ds\int_{B_{\frac{\mu^2}{2(\mu+1)}}(0)}\Big|\frac{\partial w}{\partial y_k}\Big|^2\zeta_j^2\Big(\frac{y}{\gamma}+P_j\Big)
\\& =2\gamma^{4-N}\ds\int_{\R^N}\Big|\frac{\partial w}{\partial y_k}\Big|^2
 +2\gamma^{4-N}\ds\int_{\R^N}\Big|\frac{\partial w}{\partial y_k}\Big|^2\Big[\zeta_j^2\Big(\frac{y}{\gamma}+P_j\Big)-1\Big]\\
 &\,\,\,\,\,-2\gamma^{4-N}\ds\int_{\R^N\setminus B_{\frac{\mu^2}{2(\mu+1)}}(0)}\Big|\frac{\partial w}{\partial y_k}\Big|^2\zeta_j^2\Big(\frac{y}{\gamma}+P_j\Big)\\
 &=2\gamma^{4-N}\ds\int_{\R^N}\Big|\frac{\partial w}{\partial y_k}\Big|^2+O(e^{-\sigma\mu}),
\end{align*}
  since
\begin{align*}
&\ds\int_{\R^N}\Big|\frac{\partial w}{\partial y_k}\Big|^2\Big[\zeta_j^2\Big(\frac{y}{\gamma}+P_j\Big)-1\Big]
=\ds\int_{\R^N\setminus B_{\frac{\mu-1}{2}}(0)}\Big|\frac{\partial w}{\partial y_k}\Big|^2\Big[\zeta^2\Big(\frac{2|y|}{\gamma(\mu-1)}\Big)-1\Big]\\
&\leq C\ds\int_{\R^N\setminus B_{\frac{\mu-1}{2}}(0)}\Big|\frac{\partial w}{\partial y_k}\Big|^2
\leq C\ds\int_{\R^N\setminus B_{\frac{\mu-1}{2}}(0)}e^{-2|y|}
\leq Ce^{-\sigma \mu}
\end{align*}
and similarly,
$$\ds\int_{\R^N\setminus B_{\frac{\mu^2}{2(\mu+1)}}(0)}\Big|\frac{\partial w}{\partial y_k}\Big|^2\zeta_j^2\Big(\frac{y}{\gamma}+P_j\Big)
\leq Ce^{-\sigma \mu}$$
for some $\sigma>0$.

 On the other hand, from \eqref{linfinity}, we have
\begin{equation}\label{2.4}
\Big|\ds\int_{\R^N}\Big(
\begin{array}{ccccc}
h_1 \\[1mm]
h_2\\
\end{array}
\Big) \Big(
\begin{array}{ccccc}
D_{jk,1} \\[1mm]
D_{jk,2}\\
\end{array}
\Big) \Big|
\leq \Big|\ds\int_{\R^N}h_1D_{jk,1}\Big|+\Big|\ds\int_{\R^N}h_2D_{jk,2}\Big|\
\leq C\|h\|_*.
\end{equation}

Setting $\tilde{D}_{jk}=\Big(
\begin{array}{ccccc}
\frac{\partial U_{P_j}}{\partial x_k} \\[1mm]
\frac{\partial V_{P_j}}{\partial x_k}\\
\end{array}
\Big)$, and noting that $\tilde{D}_{jk,1}=\tilde{D}_{jk,2}$, we have
\begin{equation}\label{2.5}
\arraycolsep=1.5pt
\begin{array}{rl}\displaystyle
&\,\,\,\,\,\,\,\,\ds\int_{\R^N}L\Big(
\begin{array}{ccccc}
\phi \\[1mm]
\psi\\
\end{array}
\Big) \Big(
\begin{array}{ccccc}
D_{jk,1} \\[1mm]
D_{jk,2}\\
\end{array}
\Big)=\ds\int_{\R^N}L\Big(
\begin{array}{ccccc}
D_{jk,1}  \\[1mm]
D_{jk,2} \\
\end{array}
\Big) \Big(
\begin{array}{ccccc}
\phi \\[1mm]
\psi\\
\end{array}
\Big)\\[5mm]
&=\ds\int_{\R^N}\Big(
\begin{array}{ccccc}
(\Delta\tilde{D}_{jk,1}+(\beta-1)\tilde{D}_{jk,1}
+3U_{P_j}^2\tilde{D }_{jk,1})\zeta_j \\[1mm]
(\Delta\tilde{D }_{jk,2}+(\beta-1)\tilde{D }_{jk,2}
+3V_{P_j}^2\tilde{D }_{jk,2})\zeta_j\\
\end{array}
\Big) \Big(
\begin{array}{ccccc}
\phi \\[1mm]
\psi\\
\end{array}
\Big)\\[5mm]
&\,\,\,\,\,\,~~~+\ds\int_{\R^N}\Big(
\begin{array}{ccccc}
\Delta\zeta_j\tilde{D }_{jk,1}+2\nabla\tilde{D }_{jk,1}\cdot\nabla\zeta_j\\[1mm]
\Delta\zeta_j\tilde{D }_{jk,2}+2\nabla\tilde{D }_{jk,2}\cdot\nabla\zeta_j\\
\end{array}
\Big) \Big(
\begin{array}{ccccc}
\phi \\[1mm]
\psi\\
\end{array}
\Big)+\ds\int_{\R^N}\Big(
\begin{array}{ccccc}
3U_{\mathbf{P}_m}^2\tilde{D }_{jk,1}\zeta_j-3U_{P_j}^2\tilde{D }_{jk,1}\zeta_j \\[1mm]
3V_{\mathbf{P}_m}^2\tilde{D }_{jk,2}\zeta_j-3V_{P_j}^2\tilde{D }_{jk,2}\zeta_j\\
\end{array}
\Big) \Big(
\begin{array}{ccccc}
\phi \\[1mm]
\psi\\
\end{array}
\Big)\\[5mm]
&\,\,\,\,\,\,~~~-\epsilon\ds\int_{\R^N}\Big(
\begin{array}{ccccc}
P(x)\tilde{D }_{jk,1}\zeta_j\\[1mm]
Q(x)\tilde{D }_{jk,1}\zeta_j\\
\end{array}
\Big) \Big(
\begin{array}{ccccc}
\phi \\[1mm]
\psi\\
\end{array}
\Big).
\end{array}
\end{equation}
Now we estimate all the terms in the right side of \eqref{2.5}. Firstly, since $(U_{P_j},V_{P_j})$ satisfies \eqref{1.9}, we find that the first term is equal to $0$. The second term can be
estimated as follows:
\begin{align*}
~~~&\Big|\ds\int_{\R^N}\Big(
\begin{array}{ccccc}
\Delta\zeta_j\tilde{D }_{jk,1}+2\nabla\tilde{D }_{jk,1}\cdot\nabla\zeta_j\\[1mm]
\Delta\zeta_j\tilde{D }_{jk,2}+2\nabla\tilde{D }_{jk,2}\cdot\nabla\zeta_j\\
\end{array}
\Big) \Big(
\begin{array}{ccccc}
\phi \\[1mm]
\psi\\
\end{array}
\Big)\Big|\\[5mm]
&=\Big|\ds\int_{B_{\frac{\mu^2}{2(\mu+1)}}(0)\setminus B_{\frac{\mu-1}{2}}(0)}\Big(
\begin{array}{ccccc}
\Delta\zeta_j\big(\frac{y}{\gamma}+P_j\big)\frac{\partial w}{\partial y_k}
+2\nabla \frac{\partial w}{\partial y_k}\nabla \zeta_i\big(\frac{y}{\gamma}+P_j\big)\\[1mm]
\Delta\zeta_j\big(\frac{y}{\gamma}+P_j\big) \frac{\partial w}{\partial y_k}
+2\nabla \frac{\partial w}{\partial y_k}\nabla \zeta_i\big(\frac{y}{\gamma}+P_j\big)\\
\end{array}
\Big) \Big(
\begin{array}{ccccc}
\phi \\[1mm]
\psi\\
\end{array}
\Big)\gamma^{4-N}\Big|\\&\leq Ce^{-\frac{1}{2}\sigma\mu}\|(\phi,\psi)\|_*,
\end{align*}
since
\begin{align*}
~~~&\Big|\ds\int_{B_{\frac{\mu^2}{2(\mu+1)}}(0)\setminus B_{\frac{\mu-1}{2}}(0)}
\Delta\zeta_j\Big(\frac{y}{\gamma}+P_j\Big) \frac{\partial w}{\partial y_k}\phi
+2\nabla\frac{\partial w}{\partial y_k}\cdot\nabla \zeta_j\Big(\frac{y}{\gamma}+P_j\Big)\phi
\Big|\\
&\leq C \ds\sup_{x\in\R^N}|\phi F^{-1}|
\Big|\ds\int_{B_{\frac{\mu^2}{2(\mu+1)}}(0)\setminus B_{\frac{\mu-1}{2}}(0)}
\ds\sum_{l=1}^{m}e^{-\nu\gamma|\frac{y}{\gamma}+P_j-P_l|}e^{-|y|}\Big|\\
&\leq C\ds\sup_{x\in\R^N}|\phi F^{-1}|e^{-\frac{1}{2}\eta\mu}\ds\int_{B_{\frac{\mu^2}{2(\mu+1)}}(0)\setminus B_{\frac{\mu-1}{2}}(0)}
e^{-|y|}\leq Ce^{-\frac{1}{2}\sigma\mu}\ds\sup_{x\in\R^N}|\phi F^{-1}|\\
\end{align*}
and
$$\Big|\ds\int_{B_{\frac{\mu^2}{2(\mu+1)}}(0)\setminus B_{\frac{\mu-1}{2}}(0)}
\Delta\zeta_j\Big(\frac{y}{\gamma}+P_j\Big) w_{y_k}\psi
+2\nabla w_{y_k}\cdot\nabla \zeta_j\Big(\frac{y}{\gamma}+P_j\Big)\psi
\Big|
\leq Ce^{-\frac{1}{2}\sigma\mu}\ds\sup_{x\in\R^N}|\psi F^{-1}|,$$
for some $\sigma>0$.
Similarly, we can deduce
\begin{align*}
\Big|\epsilon\ds\int_{\R^N}\Big(
\begin{array}{ccccc}
P(x)\tilde{D }_{jk,1}\zeta_j\\[1mm]
Q(x)\tilde{D }_{jk,1}\zeta_j\\
\end{array}
\Big) \Big(
\begin{array}{ccccc}
\phi \\[1mm]
\psi\\
\end{array}
\Big)\Big|\leq Ce^{-\frac{1}{2}\sigma\mu}\|(\phi,\psi)\|_*
\end{align*}
and
\begin{align*}
~~~&\Big|\ds\int_{\R^N}\Big(
\begin{array}{ccccc}
3U_{\mathbf{P}_m}^2\tilde{D }_{jk,1}\zeta_j-3U_{P_j}^2\tilde{D }_{jk,1}\zeta_j \\[1mm]
3V_{\mathbf{P}_m}^2\tilde{D }_{jk,2}\zeta_j-3V_{P_j}^2\tilde{D }_{jk,2}\zeta_j\\
\end{array}
\Big) \Big(
\begin{array}{ccccc}
\phi \\[1mm]
\psi\\
\end{array}
\Big)\Big|\\
&\leq C\ds\int_{B_{\frac{\mu}{2}}(P_j)}\Big(
\begin{array}{ccccc}
U_{P_j}\sum_{k\neq j}U_{P_k} \\[1mm]
V_{P_j}\sum_{k\neq j}V_{P_k}\\
\end{array}
\Big) \Big(
\begin{array}{ccccc}
\phi \\[1mm]
\psi\\
\end{array}
\Big)\leq Ce^{-\frac{1}{2}\sigma\mu}\|(\phi,\psi)\|_*
\end{align*}
for some $\sigma>0$. So we can conclude that
\begin{equation}\label{2.6}
|c_{jk}|\leq C(e^{-\frac{1}{2}\sigma\mu}\|(\phi,\psi)\|_*+\|h\|_*).
\end{equation}
Let now $\vartheta\in(0,1)$. It is easy to check that the function $F$ satisfies
$$
L\Big(
\begin{array}{ccccc}
F \\[1mm]
F\\
\end{array}
\Big)\leq \ds\frac{1}{2}(\vartheta^{2}-1)\Big(
\begin{array}{ccccc}
F \\[1mm]
F\\
\end{array}
\Big),
$$
in $\R^N\setminus\ds\cup_{j}^mB(P_j,\mu_1/\gamma)$ if $\mu_1$ is large enough but independent of $\mu$.
Hence the function $F$ can be used as a barrier to prove the pointwise estimate (similar to (3.11) in \cite{awz})
\begin{equation}\label{2.7}
|(\phi,\psi)|\leq C\Big(\Big\|L\Big(
\begin{array}{ccccc}
\phi \\[1mm]
\psi\\
\end{array}
\Big)\Big\|_*+\ds\sup_j\|(\phi,\psi)\|_{L^{\infty}(\partial B(P_j,\mu_1/\gamma))}\Big)F(x),
\end{equation}
for all $x\in\R^N\setminus\ds\cup_{j}^mB(P_j,\mu_1/\gamma)$.

Now we assume that there exist a sequence $\{\mu^n\}$ tending to $\infty$ and sequences $\{h^n\}$,
$\Big(
\begin{array}{ccccc}
\phi^n \\[1mm]
\psi^n\\
\end{array}
\Big)$, $\{c_{jk}^{n}\}$ such that
$$\|h^n\|_*\rightarrow0,\,\,\,\text{and}\,\,\|(\phi^n,\psi^n)\|_*=1.$$
By \eqref{2.6}, we can get
$$\Big\|\ds\sum_{j,k}c_{jk}^{n}D_{jk}\Big\|_*\rightarrow0.$$
Then \eqref{2.7} implies that there exists $P_j^n\in \Omega_m$ such that
\begin{equation}\label{2.8}
\|(\phi^n,\psi^n)\|_{L^{\infty}( B(P_j^n,\mu^n/2\gamma))}\geq C,
\end{equation}
for some constant $C>0$. Using elliptic estimates with Ascoli-Arzela's theorem, we can find a subsequence of $\{P_j^n\}$
and we can extract, from the sequence $(\phi^n(\cdot-P_j^n),\psi^n(\cdot-P_j^n))$ a subsequence which will
converge (on compact sets) to $(\phi_\infty,\psi_\infty)$ a solution of
\begin{equation}\label{2.9}\left\{%
\begin{array}{ll}
    \Delta \phi_\infty+3U^2\phi_\infty+(\beta-1)\psi_\infty=0, & \hbox{$\text{in}\,\,\R^N$}, \vspace{0.2cm}\\
    \Delta \psi_\infty+3V^2\psi_\infty+(\beta-1)\phi_\infty=0, & \hbox{$\text{in}\,\,\R^N$}.
\end{array}%
\right.\end{equation}
Moreover, recall that $(\phi^n,\psi^n)$ satisfies the orthogonal condition in \eqref{2.2e}. So,
\begin{equation}\label{2.10}
\ds\int_{\R^N}\left(\phi_\infty\nabla U+\psi_\infty\nabla V\right)=0.
\end{equation}
By the non-degeneracy of $(U,V)$, we have $(\phi_\infty,\psi_\infty)\equiv(0,0)$, which contradicts to \eqref{2.8}.
The proof is complete.

\end{proof}

Applying Proposition \ref{prop2.1}, we get the following result at once.

\begin{prop}\label{prop2.2}
Given $0<\nu<1$, there exist positive numbers $\epsilon_0$, $\mu_0$, $C\geq0$ such that for all $0<\epsilon<\epsilon_0$, $\mu\geq \mu_0$ and for
any given $h$ with $\|h\|_*$  norm bounded, there is a unique solution $\Big(\Big(
\begin{array}{ccccc}
\phi \\
\psi\\
\end{array}
\Big),\{c_{jk}\}\Big)$ to problem \eqref{2.2e}. Furthermore,
\begin{equation}\label{2.11}
\|(\phi,\psi)\|_*\leq C\|h\|_*.
\end{equation}

\end{prop}

\begin{proof}
Consider the space
$$\mathcal{H}=\Big\{(u,v)\in \mathbf{H}\,\Big|\,\,\Big \langle  \Big(
\begin{array}{ccccc}
u \\[1mm]
v\\
\end{array}
\Big),\Big(
\begin{array}{ccccc}
D_{jk,1} \\[1mm]
D_{jk,2}\\
\end{array}
\Big)\Big\rangle=0 ,\,\,\mathbf{P}_m\in \Omega_m\Big\}.$$
Since the problem \eqref{2.2e} can be rewritten as
\begin{equation}\label{2.12}
\Big(
\begin{array}{ccccc}
\phi \\[1mm]
\psi\\
\end{array}
\Big)+\Big(
\begin{array}{ccccc}
\mathcal{K}_1\ & \mathcal{K}\\ [2mm]
\mathcal{K} \ & \mathcal{K}_2 \\
\end{array}
\Big)\Big(
\begin{array}{ccccc}
\phi \\[1mm]
\psi\\
\end{array}
\Big)=\bar{h}\,\,\,\text{in}\,\,\mathcal{H},
\end{equation}
where $\overline{h}$ is defined by duality and $\mathcal{K},\mathcal{K}_1,\mathcal{K}_2:\mathcal{H}\rightarrow\mathcal{H}$ are linear
compact operators. By Fredholm's alternative theorem, we know that \eqref{2.12} has a unique solution for each
$\bar{h}$ is equivalent to showing that the system has a unique solution for $\bar{h}=0$,
which in turn  follows from Proposition \ref{prop2.1}. This concludes the proof of Proposition \ref{prop2.2}.

\end{proof}
In the following, if $(\phi,\psi)$ is the unique solution given by Proposition \ref{prop2.2}, we denote
\begin{equation}\label{2.13}
(\phi,\psi)=\mathcal{A}(h)
\end{equation}
and \eqref{2.11} yields
\begin{equation}\label{2.14}
\|\mathcal{A}(h)\|_*\leq C\|h\|_*.
\end{equation}

Now we reduce \eqref{1.1} to a finite-dimensional one.
For large $\mu$ and fixed $\textbf{P}_m\in \Omega_m$,
we are going to find a function $\{(\phi_{\textbf{P}_m},\psi_{\textbf{P}_m})\}$ such that for some $\{c_{jk}\},j=1,\cdots,m,k=1,\cdots,N$,
the following nonlinear projected problem holds true
\begin{equation}\label{3.1}\left\{%
\begin{array}{ll}
\Big(
\begin{array}{ccccc}
\Delta (U_{\textbf{P}_m}+\phi_{\textbf{P}_m})-(1+\epsilon P(x))(U_{\textbf{P}_m}+\phi_{\textbf{P}_m})+
 (U_{\textbf{P}_m}+\phi_{\textbf{P}_m})^3
 +\beta(V_{\textbf{P}_m}+\psi_{\textbf{P}_m}) \\[1mm]
\Delta (V_{\textbf{P}_m}+\psi_{\textbf{P}_m})-(1+\epsilon Q(x))(V_{\textbf{P}_m}+\psi_{\textbf{P}_m})+
 (V_{\textbf{P}_m}+\psi_{\textbf{P}_m})^3+\beta(U_{\textbf{P}_m}+\phi_{\textbf{P}_m}) \\
\end{array}
\Big)\\[4mm]
=\Big(
\begin{array}{ccccc}
\sum_{j=1}^{m}\sum_{k=1}^{N}c_{jk}D_{jk,1},\vspace{0.2cm}\\
\sum_{i=1}^{m}\sum_{k=1}^{N}c_{jk}D_{jk,2}\\
\end{array}
\Big),\vspace{0.2cm}\\
\Big \langle  \Big(
\begin{array}{ccccc}
\phi_{\textbf{P}_m} \\[1mm]
\psi_{\textbf{P}_m}\\
\end{array}
\Big),\Big(
\begin{array}{ccccc}
D_{jk,1} \\[1mm]
D_{jk,2}\\
\end{array}
\Big)\Big\rangle=0 \,\,\,\text{for}\,\,j=1,\cdots,m,\,k=1,\cdots,N.
\end{array} %
\right.
\end{equation}

It is obvious that the first system in \eqref{3.1} can be rewritten as
\begin{equation}\label{3.2}
L\Big(
\begin{array}{ccccc}
\phi_{\textbf{P}_m} \\[1mm]
\psi_{\textbf{P}_m}\\
\end{array}
\Big)=-G\Big(
\begin{array}{ccccc}
U_{\textbf{P}_m} \\[1mm]
V_{\textbf{P}_m}\\
\end{array}
\Big)+M\Big(
\begin{array}{ccccc}
\phi_{\textbf{P}_m} \\[1mm]
\psi_{\textbf{P}_m}\\
\end{array}
\Big)+\ds\sum_{j=1}^{m}\sum_{k=1}^{N}c_{jk}\Big(
\begin{array}{ccccc}
D_{jk,1} \\[1mm]
D_{jk,2}\\
\end{array}
\Big),
\end{equation}
where\begin{equation}\label{3.3}
M\Big(
\begin{array}{ccccc}
\phi_{\textbf{P}_m} \\[1mm]
\psi_{\textbf{P}_m}\\
\end{array}
\Big)=\Big(
\begin{array}{ccccc}
(U_{\textbf{P}_m}+\phi_{\textbf{P}_m})^3-U_{\textbf{P}_m}^3-3U_{\textbf{P}_m}^2\phi_{\textbf{P}_m}\\[1mm]
(V_{\textbf{P}_m}+\psi_{\textbf{P}_m})^3-V_{\textbf{P}_m}^3-3V_{\textbf{P}_m}^2\psi_{\textbf{P}_m}\\
\end{array}
\Big).
\end{equation}

Now we come to the main result in this section.
 \begin{prop}\label{prop3.3}
There exist positive numbers $\mu_0$, $C$ and $\sigma>0$ such that for all $\mu>\mu_0$, and
for any $\textbf{P}_m\in \Omega_m$, $\epsilon<e^{-2\mu}$, there is a unique solution $\Big(\Big(
\begin{array}{ccccc}
\phi_{\textbf{P}_m} \\[1mm]
\psi_{\textbf{P}_m}\\
\end{array}
\Big),\{c_{jk}\}\Big)$ to problem \eqref{3.1}. Furthermore, $(\phi_{\textbf{P}_m},\psi_{\textbf{P}_m})$
is $C^1\times C^1$ in $\Omega_m$ and
\begin{equation}\label{3.6}
\|(\phi_{\textbf{P}_m},\psi_{\textbf{P}_m})\|_*\leq Ce^{-\sigma\mu},\,\,\,
|c_{j,k}|\leq Ce^{-\sigma\mu}.
\end{equation}
\end{prop}

In order to apply the contraction theorem to prove Proposition \ref{prop3.3}, firstly we have to obtain the following two lemmas.

 \begin{lem}\label{lem3.1}
Assume that $0<\nu<1$. For $\mu$ large enough, and any $\textbf{P}_m\in \Omega_m$, $\epsilon<e^{-2\mu}$, we have
\begin{equation}\label{3.4}
\Big\|G\Big(
\begin{array}{ccccc}
U_{\textbf{P}_m} \\[1mm]
V_{\textbf{P}_m}\\
\end{array}
\Big)\Big\|_*\leq Ce^{-\sigma\mu}
\end{equation}
for some constants $\sigma$ and $C$ independent of $\mu$, $m$ and $\textbf{P}_m$.
\end{lem}

\begin{proof}
Using the system \eqref{1.9} satisfied by $(U_{P_j},V_{P_j})$, $j=1,\cdots,m$ , we have
\begin{align*}
G\Big(
\begin{array}{ccccc}
U_{\textbf{P}_m} \\[1mm]
V_{\textbf{P}_m}\\
\end{array}
\Big)&=\Big(
\begin{array}{ccccc}
\Delta U_{\textbf{P}_m}-(1+\epsilon P(x))U_{\textbf{P}_m}+U_{\textbf{P}_m}^3+\beta V_{\textbf{P}_m} \\[1mm]
\Delta V_{\textbf{P}_m}-(1+\epsilon Q(x))V_{\textbf{P}_m}+V_{\textbf{P}_m}^3+\beta U_{\textbf{P}_m}
\end{array}
\Big)\\
&=\Big(
\begin{array}{ccccc}
U_{\textbf{P}_m}^3-\sum_{j=1}^{m}U_{P_j}^3 \\[1mm]
V_{\textbf{P}_m}^3-\sum_{j=1}^{m}V_{P_j}^3
\end{array}
\Big)-\Big(
\begin{array}{ccccc}
\epsilon P(x)U_{\textbf{P}_m} \\[1mm]
\epsilon Q(x)V_{\textbf{P}_m}
\end{array}
\Big).
\end{align*}
Fix $k\in\{1,\cdots,m\}$ and consider the region $|x-P_k|\leq\mu/2\gamma$. In this region, we have
\begin{align*}
\Big|U_{\textbf{P}_m}^3-\ds\sum_{j=1}^mU_{P_j}^3\Big|
&=\Big|U_{\textbf{P}_m}^3-U_{P_k}^3-\ds\sum_{j\neq k}U_{P_j}^3\Big|\leq C\Big|U_{P_k}^2\ds\sum_{j\neq k}U_{P_j}+\ds\sum_{j\neq k}U_{P_j}^3\Big|\\
&\leq C\big(e^{-\frac{1}{2}\mu}e^{-2\gamma|x-P_k|}+e^{-\frac{3}{2}\mu}\big)
\leq Ce^{-\sigma\mu}e^{-\nu\gamma|x-P_k|}
\end{align*}
and similarly,$$
\Big|V_{\textbf{P}_m}^3-\ds\sum_{j=1}^{m}V_{P_j}^3\Big|
\leq Ce^{-\sigma\mu}e^{-\nu\gamma|x-P_k|}
$$
for a proper choice of $\sigma>0$.
Consider the region $|x-P_k|>\mu/2\gamma$ for all $k\in\{1,\cdots,m\}$. We have
\begin{align*}
&\Big|U_{\textbf{P}_m}^3-\ds\sum_{j=1}^mU_{P_j}^3\Big|
\leq C\ds\sum_{k=1}^{m}U_{P_k}^3\leq C\ds\sum_{k=1}^{m}e^{-3\gamma|x-P_k|}\\
&\leq C\ds\sum_{k=1}^{m}e^{-\nu\gamma|x-P_k|}e^{-(3-\nu)\gamma|x-P_k|}
\leq Ce^{-\sigma\mu}\ds\sum_{k=1}^{m}e^{-\nu\gamma|x-P_k|}
\end{align*}
 and
 $$
\Big|V_{\textbf{P}_m}^3-\ds\sum_{j=1}^{m}V_{P_j}^3\Big|
\leq Ce^{-\sigma\mu}\ds\sum_{k=1}^{m}e^{-\nu\gamma|x-P_k|}
$$
for a proper choice of $\sigma>0$.
Now, under the assumption on $\epsilon$, it is easy to see that
$$
\left|\epsilon P(x)U_{\textbf{P}_m}\right|
\leq Ce^{-\sigma\mu}\ds\sum_{k=1}^{m}e^{-\nu\gamma|x-P_k|}
$$
and$$
\left|\epsilon Q(x)V_{\textbf{P}_m}\right|
\leq Ce^{-\sigma\mu}\ds\sum_{k=1}^{m}e^{-\nu\gamma|x-P_k|}.
$$

Thus using the above estimates, we have
$$
\Big\|G\Big(
\begin{array}{ccccc}
U_{\textbf{P}_m} \\[1mm]
V_{\textbf{P}_m}\\
\end{array}
\Big)\Big\|_*\leq Ce^{-\sigma\mu}
$$
for some $\sigma>0$.
\end{proof}

Letting $r>0$, define $$
\mathcal{B}=\Big\{(\phi,\psi)\in \textbf{H}\,\,:\,\|(\phi,\psi)\|_*\leq re^{-\sigma\mu},\,\,\Big \langle  \Big(
\begin{array}{ccccc}
\phi \\[1mm]
\psi\\
\end{array}
\Big),\Big(
\begin{array}{ccccc}
D_{jk,1} \\[1mm]
D_{jk,2}\\
\end{array}
\Big)\Big\rangle=0\Big\}.
$$

\begin{lem}\label{lem3.2}
 For any $\textbf{P}_m\in \Omega_m$, if $(\phi,\psi)\in \mathcal{B}$, we have
 \begin{equation}\label{3.5}
\Big\|M\Big(
\begin{array}{ccccc}
\phi \\[1mm]
\psi\\
\end{array}
\Big)\Big\|_*\leq Ce^{-\sigma\mu}
\end{equation}
for some constants $\sigma$ and $C$ independent of $\mu$, $m$ and $\textbf{P}_m$.
  \end{lem}

\begin{proof}
 By direct computation and the mean-value theorem, we have
$$|(U_{\textbf{P}_m}+\phi)^3-U_{\textbf{P}_m}^3-3U_{\textbf{P}_m}^2\phi|
\leq C(|\phi|^2+|\phi|^3)$$
and$$
|(V_{\textbf{P}_m}+\psi)^3-V_{\textbf{P}_m}^3-3V_{\textbf{P}_m}^2\psi|
\leq C(|\psi|^2+|\psi|^3).$$
As a result,
\begin{align*}
\Big\|M\Big(
\begin{array}{ccccc}
\phi \\[1mm]
\psi\\
\end{array}
\Big)\Big\|_* &\leq C\|\phi\|_*(|\phi|+\phi^2)+C\|\psi\|_*(|\psi|+\psi^2)\\
&\leq C\big(\|\phi\|_*^2+\|\phi\|_*^3\big)+C(\|\psi\|_*^2+\|\psi\|_*^3)\\
&\leq C(r^2e^{-2\sigma\mu}+r^3e^{-3\sigma\mu})
\leq Ce^{-\sigma\mu}
\end{align*}
for a proper $\sigma$ independent of $\mu,m$ and $\textbf{P}_m$.

 \end{proof}

Now we are in position to prove Proposition \ref{prop3.3}.

\begin{proof}[ \textbf{Proof of Proposition \ref{prop3.3}}]
We will use the contraction mapping theorem to prove it.
Notice that $(\phi_{\textbf{P}_m},\psi_{\textbf{P}_m})$ solves \eqref{3.1} if and only if
$$
\Big(
\begin{array}{ccccc}
\phi_{\textbf{P}_m} \\[1mm]
\psi_{\textbf{P}_m}\\
\end{array}
\Big)=\mathcal{A}\Big(G\Big(
\begin{array}{ccccc}
U_{\textbf{P}_m} \\[1mm]
V_{\textbf{P}_m}\\
\end{array}
\Big)+M\Big(
\begin{array}{ccccc}
\phi_{\textbf{P}_m} \\[1mm]
\psi_{\textbf{P}_m}\\
\end{array}
\Big)\Big),
$$
where $\mathcal{A}$ is the operator given by \eqref{2.13}. In other words,
$(\phi_{\textbf{P}_m},\psi_{\textbf{P}_m})$ solves \eqref{3.1} if and only if
$(\phi_{\textbf{P}_m},\psi_{\textbf{P}_m})$ is a fixed point for the operator
$$
\mathcal{T}\Big(
\begin{array}{ccccc}
\phi \\[1mm]
\psi\\
\end{array}
\Big)=:\mathcal{A}\Big(G\Big(
\begin{array}{ccccc}
U_{\textbf{P}_m} \\[1mm]
V_{\textbf{P}_m}\\
\end{array}
\Big)+M\Big(
\begin{array}{ccccc}
\phi \\[1mm]
\psi\\
\end{array}
\Big)\Big).
$$

We will prove that $\mathcal{T}$ is a contraction mapping from $\mathcal{B}$ to itself.
On one hand, by \eqref{2.14}, Lemmas \ref{lem3.1} and \ref{lem3.2}, we have for any $(\phi,\psi)\in\mathcal{B}$,
$$
\Big\|\mathcal{T}\Big(
\begin{array}{ccccc}
\phi \\[1mm]
\psi\\
\end{array}
\Big)\Big\|_*
\leq
C\Big\|G\Big(
\begin{array}{ccccc}
U_{\textbf{P}_m} \\[1mm]
V_{\textbf{P}_m}\\
\end{array}
\Big)+M\Big(
\begin{array}{ccccc}
\phi \\[1mm]
\psi\\
\end{array}
\Big)\Big\|_*\leq Ce^{-\sigma\mu}.
$$

On the other hand, taking $\Big(
\begin{array}{ccccc}
\phi_1 \\[1mm]
\psi_1\\
\end{array}
\Big)$ and $\Big(
\begin{array}{ccccc}
\phi_2 \\[1mm]
\psi_2\\
\end{array}
\Big)$ in $\mathcal{B}$, we have
$$M\Big(
\begin{array}{ccccc}
\phi_1 \\[1mm]
\psi_1\\
\end{array}
\Big)-M\Big(
\begin{array}{ccccc}
\phi_2 \\[1mm]
\psi_2\\
\end{array}
\Big)=\Big(
\begin{array}{ccccc}
(U_{\textbf{P}_m}+\phi_1)^3-(U_{\textbf{P}_m}+\phi_2)^3-3U_{\textbf{P}_m}^2(\phi_{1}-\phi_2) \\[1mm]
(V_{\textbf{P}_m}+\psi_1)^3-(V_{\textbf{P}_m}+\psi_2)^3-3V_{\textbf{P}_m}^2(\psi_{1}-\psi_2)\\
\end{array}
\Big),$$
\begin{align*}
~~~&\big|(U_{\textbf{P}_m}+\phi_1)^3-(U_{\textbf{P}_m}+\phi_2)^3-3U_{\textbf{P}_m}^2(\phi_{1}-\phi_2)\big|\\
&=\big|(\phi_{1}-\phi_2)(3U_{\textbf{P}_m}(\phi_1+\phi_2)+\phi_1^2+\phi_1\phi_2+\phi_2^2)\big|\\
&\leq C|\phi_{1}-\phi_2|\left(|\phi_1|+|\phi_2|\right)
\end{align*}
and
$$
|(V_{\textbf{P}_m}+\psi_1)^3-(V_{\textbf{P}_m}+\psi_2)^3-3V_{\textbf{P}_m}^2(\psi_{1}-\psi_2)|
\leq C|\psi_{1}-\psi_2|(|\psi_1|+|\psi_2|).
$$
So we can infer that
\begin{align*}
\Big\|\mathcal{T}\Big(
\begin{array}{ccccc}
\phi_1 \\[1mm]
\psi_1\\
\end{array}
\Big)-\mathcal{T}\Big(
\begin{array}{ccccc}
\phi_2 \\[1mm]
\psi_2\\
\end{array}
\Big)\Big\|_* &\leq C\Big\|M\Big(
\begin{array}{ccccc}
\phi_1 \\[1mm]
\psi_1\\
\end{array}
\Big)-M\Big(
\begin{array}{ccccc}
\phi_2 \\[1mm]
\psi_2\\
\end{array}
\Big)\Big\|_*\\
&\leq C\|\phi_{1}-\phi_2\|_*(\|\phi_{1}\|_*+\|\phi_2\|_*)
+C\|\psi_{1}-\psi_2\|_*(\|\psi_{1}\|_*+\|\psi_2\|_*)\\
&\leq \frac{1}{2}(\|\phi_{1}-\phi_2\|_*+\|\psi_{1}-\psi_2\|_*)\\
&=\frac{1}{2}\Big\|\Big(
\begin{array}{ccccc}
\phi_1 \\[1mm]
\psi_1\\
\end{array}
\Big)-\Big(
\begin{array}{ccccc}
\phi_2 \\[1mm]
\psi_2\\
\end{array}
\Big)\Big\|_*.
\end{align*}
This means that $\mathcal{T}$ is a contraction mapping from $\mathcal{B}$ to itself.
By the contraction mapping theorem, there exists a unique $\Big(
\begin{array}{ccccc}
\phi_{\textbf{P}_m} \\[1mm]
\psi_{\textbf{P}_m}\\
\end{array}
\Big)\in\mathcal{B}$ such that \eqref{3.1} holds. So,
$$
\Big\|\Big(\begin{array}{ccccc}
\phi_{\textbf{P}_m} \\[1mm]
\psi_{\textbf{P}_m}\\
\end{array}
\Big)\Big\|_*=\Big\|\mathcal{T}\Big(
\begin{array}{ccccc}
\phi_{\textbf{P}_m} \\[1mm]
\psi_{\textbf{P}_m}\\
\end{array}
\Big)\Big\|_*\leq Ce^{-\sigma\mu}.
$$

Furthermore, combing \eqref{2.6}, \eqref{3.4} and \eqref{3.5}, we find
$$|c_{j,k}|\leq Ce^{-\sigma\mu}.$$

\end{proof}

\section{{A secondary Liapunov-Schmidt reduction}}\label{s5}

In this section, we prove a key estimate on the difference between the solutions in the $m$-th step and
the $(m+1)$-th step. This second Liapunov-Schmidt reduction has been used in \cite{aw,awz}.
For $\textbf{P}_{m}\in \Omega_m$, we denote
$$\Big(
\begin{array}{ccccc}
u_{\mathbf{P}_m} \\[1mm]
v_{\mathbf{P}_m}\\
\end{array}
\Big):=\Big(
\begin{array}{ccccc}
U_{\textbf{P}_m}+\phi_{\textbf{P}_m} \\[1mm]
V_{\textbf{P}_m}+\psi_{\textbf{P}_m}\\
\end{array}
\Big),$$
where $\Big(
\begin{array}{ccccc}
\phi_{\textbf{P}_m} \\[1mm]
\psi_{\textbf{P}_m}\\
\end{array}
\Big)$ is the unique solution given by Proposition \ref{prop3.3}.

We now write
\begin{equation}\label{4.1}
\arraycolsep=1.5pt
\begin{array}{rl}\displaystyle
\Big(
\begin{array}{ccccc}
u_{\mathbf{P}_{m+1}} \\[1mm]
v_{\mathbf{P}_{m+1}}\\
\end{array}
\Big)&=\Big(
\begin{array}{ccccc}
u_{\textbf{P}_m} \\[1mm]
v_{\textbf{P}_m}\\
\end{array}
\Big)+\Big(
\begin{array}{ccccc}
U_{P_{m+1}} \\[1mm]
V_{P_{m+1}}\\
\end{array}
\Big)+\varphi_{m+1}=\Big(
\begin{array}{ccccc}
\bar{U} \\[1mm]
\bar{V}\\
\end{array}
\Big)+\Big(
\begin{array}{ccccc}
\varphi_{m+1,1} \\[1mm]
\varphi_{m+1,2}\\
\end{array}
\Big),
\end{array}
\end{equation}
where $\Big(
\begin{array}{ccccc}
\bar{U} \\[1mm]
\bar{V}\\
\end{array}
\Big)=\Big(
\begin{array}{ccccc}
u_{\textbf{P}_m} \\[1mm]
v_{\textbf{P}_m}\\
\end{array}
\Big)+\Big(
\begin{array}{ccccc}
U_{P_{m+1}} \\[1mm]
V_{P_{m+1}}\\
\end{array}
\Big)$.

By Proposition \ref{prop3.3}, we can easily derive that
\begin{equation}\label{4.2}
\|(\varphi_{m+1,1},\varphi_{m+1,2})\|_*\leq Ce^{-\sigma\mu}.
\end{equation}
But the estimate is not sufficient, we need a key estimate for $\Big(
\begin{array}{ccccc}
\varphi_{m+1,1} \\[1mm]
\varphi_{m+1,2}\\
\end{array}
\Big)$ which will be given later. In the following we will always assume that $\nu>\frac{1}{2}$.

\begin{lem} \label{lem4.2}

Letting $\mu$ and $\epsilon$ be as in Proposition \ref{prop3.3}, then it holds
\begin{equation}\label{4.3}
\arraycolsep=1.5pt
\begin{array}{rl}\displaystyle
&\,\,\,\,\,\,\ds\int_{\R^N}\left(|\nabla \varphi_{m+1,1}|^2+|\varphi_{m+1,1}|^2+|\nabla \varphi_{m+1,2}|^2+|\varphi_{m+1,2}|^2\right)\\[2mm]
&\leq Ce^{-\sigma\mu}\ds\sum_{j=1}^mw(\gamma|P_{m+1}-P_j|)
+C\epsilon^2\Big[\Big(\ds\int_{\R^N}\left(|P(x)|U_{P_{m+1}}+|Q(x)|V_{P_{m+1}}\right)\Big)^2\\[2mm]
&\,\,\,\,\,\,+\ds\int_{\R^N}\left(|P(x)|^2U_{P_{m+1}}^2+|Q(x)|^2V_{P_{m+1}}^2\right)\Big]
\end{array}
\end{equation}
for some $C>0,\sigma>0$ independent of $\mu,m$ and $\mathbf{P}_{m+1}\in\Omega_{m+1}$.
\end{lem}

\begin{proof}
To prove \eqref{4.3}, we need to perform a further decomposition.
From the non-degeneracy result of $(U,V)$, we know that there are finitely many positive eigenvalues to the following
linearized operators
\begin{equation}\label{4.4}
\Big(
\begin{array}{ccccc}
\Delta \eta_{l,1}+(\beta-1) \eta_{l,1}+3U^2\eta_{l,1} \\[1mm]
\Delta \eta_{l,2}+(\beta-1) \eta_{l,2}+3V^2\eta_{l,2}
\end{array}
\Big)
=\lambda_l\Big(
\begin{array}{ccccc}
\eta_{l,1} \\[1mm]
\eta_{l,2},
\end{array}
\Big)
\end{equation}
and the eigenfunctions $\phi_l$ are exponential decay. Assume that $\lambda_l>0$ for $l=1,\cdots,K$, and it is easy to see that
$\eta_{l,1} =\eta_{l,2} $.
Let $\omega_{jl}=\zeta_j\eta_l(\gamma(x-P_j))$, where $\zeta_j$ is given in Section \ref{s3} and $\eta_l=\Big(
\begin{array}{ccccc}
\eta_{l,1} \\[1mm]
\eta_{l,2}
\end{array}
\Big)$.

It follows from \eqref{4.1} that
 $$
 \varphi_{m+1}=\Big(
\begin{array}{ccccc}
u_{\mathbf{P}_{m+1}} \\[1mm]
v_{\mathbf{P}_{m+1}}\\
\end{array}
\Big)-\Big(
\begin{array}{ccccc}
u_{\textbf{P}_m} \\[1mm]
v_{\textbf{P}_m}\\
\end{array}
\Big)-\Big(
\begin{array}{ccccc}
U_{P_{m+1}} \\[1mm]
V_{P_{m+1}}\\
\end{array}
\Big)
$$
and then
\begin{equation}\label{4.5}
\bar{L}\varphi_{m+1}
=-\bar{G}+\ds\sum_{j=1}^{m+1}\ds\sum_{k=1}^{N}c_{jk}D_{jk}\,\,\text{for some} \,c_{jk},
\end{equation}
where
$$
\bar{L}\Big(
\begin{array}{ccccc}
\varphi_{m+1,1} \\[1mm]
\varphi_{m+1,2}
\end{array}
\Big)=\Big(
\begin{array}{ccccc}
\Delta \varphi_{m+1,1}-(1+\epsilon P(x))\varphi_{m+1,1}+3\widetilde{U}^2\varphi_{m+1,1}+\beta \varphi_{m+1,2} \\[1mm]
\Delta \varphi_{m+1,2}-(1+\epsilon Q(x))\varphi_{m+1,2}+3\widetilde{V}^2\varphi_{m+1,2}+\beta \varphi_{m+1,1}
\end{array}
\Big),
$$

$$
3\widetilde{U}^2=\left\{%
\begin{array}{ll}
   \ds\frac{(\bar{U}+\varphi_{m+1,1})^3-\bar{U}^3}{\varphi_{m+1,1}},\,\,&\text{if}\,\,\varphi_{m+1,1}\neq0, \\[5mm]
    3\bar{U}^2,\,\,&\text{if}\,\,\varphi_{m+1,1}=0,
\end{array}%
\right.$$

$$3\widetilde{V}^2=\left\{%
\begin{array}{ll}
   \ds\frac{(\bar{V}+\varphi_{m+1,2})^3-\bar{V}^3}{\varphi_{m+1,2}},\,\,&\text{if}\,\,\varphi_{m+1,2}\neq0, \\[5mm]
    3\bar{V}^2,\,\,&\text{if}\,\,\varphi_{m+1,2}=0
\end{array}%
\right.
$$
and
$$
\bar{G}=\Big(
\begin{array}{ccccc}
\bar{U}^3-u_{\textbf{P}_m}^3 -U_{P_{m+1}}^3\\[1mm]
\bar{V}^3-v_{\textbf{P}_m}^3-V_{P_{m+1}}^3\\
\end{array}
\Big)-\epsilon\Big(
\begin{array}{ccccc}
P(x)U_{P_{m+1}} \\[1mm]
Q(x)V_{P_{m+1}}\\
\end{array}
\Big).
$$
We proceed the proof into a few steps.
First we estimate the $L^2$ norm of $\bar{G}$. Notice that
\begin{equation}\label{4.6}
\arraycolsep=1.5pt
\begin{array}{rl}\displaystyle
&\,\,\,\,\ds\int_{\R^N}\big|\bar{U}^3-u_{\textbf{P}_m}^3 -U_{P_{m+1}}^3\big|^2
\leq C\ds\int_{\R^N}\left(u_{\textbf{P}_m}^4 U_{P_{m+1}}^2+u_{\textbf{P}_m}^2 U_{P_{m+1}}^4\right)\\[5mm]
&\leq C\ds\int_{\R^N}\left(U_{\textbf{P}_m}^4 U_{P_{m+1}}^2+\phi_{\textbf{P}_m}^4 U_{P_{m+1}}^2+
U_{\textbf{P}_m}^2 U_{P_{m+1}}^4+\phi_{\textbf{P}_m}^2 U_{P_{m+1}}^4\right).
\end{array}
\end{equation}
Then similar to \cite{aw}, we have
\begin{equation}\label{4.7}
\arraycolsep=1.5pt
\begin{array}{rl}\displaystyle
\|\bar{G}\|_{L^2(\R^N)}^2&\leq Ce^{-\sigma\mu}\ds\sum_{j=1}^mw(\gamma|P_{m+1}-P_j|)+
\epsilon^2\ds\int_{\R^N}(|P(x)|^2U_{P_{m+1}}^2+|Q(x)|^2V_{P_{m+1}}^2).
\end{array}
\end{equation}

 Now we decompose $\varphi_{m+1}$ as
 \begin{equation}\label{4.8}
\arraycolsep=1.5pt
\begin{array}{rl}\displaystyle
\varphi_{m+1}=\Phi+\ds\sum_{j=1}^{m+1}\ds\sum_{l=1}^{K}r_{jl}\omega_{jl}
+\ds\sum_{j=1}^{m+1}\ds\sum_{k=1}^{N}d_{jk}D_{jk}
\end{array}
\end{equation}
for some $r_{jl},d_{jk}$ such that
 \begin{equation}\label{4.9}
\left \langle  \Phi,\omega_{jl}\right\rangle=\left \langle  \Phi,D_{jk}\right\rangle=0,
\,\,j=1,\cdots,m+1,k=1,\cdots,N,l=1,\cdots,K.
\end{equation}
Since
 $$
\varphi_{m+1}=\Big(
\begin{array}{ccccc}
\phi_{\textbf{P}_{m+1}}\\[1mm]
\psi_{\textbf{P}_{m+1}}\\
\end{array}
\Big)-\Big(
\begin{array}{ccccc}
\phi_{\textbf{P}_{m}}\\[1mm]
\psi_{\textbf{P}_{m}},\\
\end{array}
\Big),
$$
we have for $j=1,\cdots,m$,
\begin{align*}
d_{jk}&=\left \langle  \varphi_{m+1},D_{jk}\right\rangle+
\ds\sum_{l=1}^{K}r_{jl}\left \langle  \omega_{jl},D_{jk}\right\rangle\\
&=\Big \langle  \Big(
\begin{array}{ccccc}
\phi_{\textbf{P}_{m+1}}\\[1mm]
\psi_{\textbf{P}_{m+1}}\\
\end{array}
\Big)-\Big(
\begin{array}{ccccc}
\phi_{\textbf{P}_{m}}\\[1mm]
\psi_{\textbf{P}_{m}}\\
\end{array}
\Big),D_{jk}\Big\rangle+
\ds\sum_{l=1}^{K}r_{jl}\left \langle  \omega_{jl},D_{jk}\right\rangle\\
&=\ds\sum_{l=1}^{K}r_{jl}\left \langle  \omega_{jl},D_{jk}\right\rangle
=Ce^{-\sigma\mu}\ds\sum_{l=1}^{K}r_{jl},
\end{align*}
where we use the fact that
\begin{align*}
\left\langle  \omega_{jl},D_{jk}\right\rangle&=\ds\int_{\R^N}\Big(
\begin{array}{ccccc}
\omega_{jl,1}\\[1mm]
\omega_{jl,2}\\
\end{array}
\Big) \Big(
\begin{array}{ccccc}
D_{jk,1}\\[1mm]
D_{jk,2}\\
\end{array}
\Big)\leq C\ds\int_{B_{\frac{\mu}{2\gamma}}(P_j)}\zeta_j^2\eta_{l,1}(\gamma(x-P_j))\frac{\partial U_{P_j}}{\partial x_k}\\
&\leq C\ds\int_{B_{\frac{\mu}{2\gamma}}(P_j)}e^{-2\gamma|x-P_j|}\leq Ce^{-\sigma\mu}.
\end{align*}

For $j=m+1$, there holds
\begin{align*}
d_{m+1,k}=&\left \langle  \varphi_{m+1},D_{m+1,k}\right\rangle+
\ds\sum_{l=1}^{K}r_{m+1,l}\left \langle  \phi_{m+1,l},D_{m+1,k}\right\rangle\\
&=\Big \langle \Big(
\begin{array}{ccccc}
\phi_{\textbf{P}_{m+1}}\\[1mm]
\psi_{\textbf{P}_{m+1}}\\
\end{array}
\Big)-\Big(
\begin{array}{ccccc}
\phi_{\textbf{P}_{m}}\\[1mm]
\psi_{\textbf{P}_{m}}\\
\end{array}
\Big),D_{m+1,k}\Big\rangle+
\ds\sum_{l=1}^{K}r_{m+1,l}\left \langle  \omega_{m+1,l},D_{m+1,k}\right\rangle\\
&=-\Big \langle  \Big(
\begin{array}{ccccc}
\phi_{\textbf{P}_{m}}\\[1mm]
\psi_{\textbf{P}_{m}}\\
\end{array}
\Big),D_{m+1,k}\Big\rangle+\ds\sum_{l=1}^{K}r_{m+1,l}\left \langle  \omega_{m+1,l},Z_{m+1,k}\right\rangle,
\end{align*}
where we used the orthogonality conditions satisfied by $\Big(
\begin{array}{ccccc}
\phi_{\textbf{P}_{m}}\\[1mm]
\psi_{\textbf{P}_{m}}\\
\end{array}
\Big)$ and $ \Big(
\begin{array}{ccccc}
\phi_{\textbf{P}_{m+1}}\\[1mm]
\psi_{\textbf{P}_{m+1}}\\
\end{array}
\Big)$.

By the definition of $D_{jk}$, we have
 \begin{align*}
\Big \langle  \Big(
\begin{array}{ccccc}
\phi_{\textbf{P}_{m}}\\[1mm]
\psi_{\textbf{P}_{m}}\\
\end{array}
\Big),D_{m+1,k}\Big\rangle
&=\ds\int_{\R^N}\Big(\phi_{\textbf{P}_{m}}\zeta_{m+1}\frac{\partial U_{P_{m+1}}}{\partial x_k}
+\psi_{\textbf{P}_{m}}\zeta_{m+1}\frac{\partial V_{P_{m+1}}}{\partial x_k}\Big)\\
&\leq C e^{-\sigma\mu}\ds\int_{\R^N}\ds\sum_{j=1}^{m}e^{-\nu\gamma|x-P_j|}e^{-\nu\gamma|x-P_{m+1}|}e^{-(1-\nu)\gamma|x-P_{m+1}|}\\
&\leq C e^{-\sigma\mu}\ds\sum_{j=1}^me^{-\nu\gamma|P_{m+1}-P_j|}.
\end{align*}
So we can deduce that
\begin{equation}\label{4.10}\left\{%
\begin{array}{ll}
    |d_{m+1,k}|\leq Ce^{-\sigma\mu}\ds\sum_{j=1}^me^{-\nu\gamma|P_{m+1}-P_j|}+Ce^{-\sigma\mu}\ds\sum_{l=1}^{K}r_{m+1,l}, \\[3mm]
   |d_{jk}|\leq Ce^{-\sigma\mu}\ds\sum_{l=1}^{K}r_{jl}\,\,\text{for}\,j=1,\cdots,m.
\end{array}%
\right.\end{equation}

It follows from \eqref{4.8} that \eqref{4.5} can be rewritten as
\begin{equation}\label{4.11}
\bar{L}\Phi+\ds\sum_{j=1}^{m+1}\ds\sum_{l=1}^{K}r_{jl}\bar{L}\omega_{jl}
+\ds\sum_{j=1}^{m+1}\ds\sum_{k=1}^{N}d_{jk}\bar{L}D_{jk}=-\bar{G}+\ds\sum_{j=1}^{m+1}\ds\sum_{k=1}^{N}c_{jk}D_{jk}.
\end{equation}
To estimate the coefficients $r_{jl}$, $l\in\{1,\cdots K\}$, multiplying \eqref{4.11} by $\omega_{jl}$ and integrating over $\R^N$,
we have
\begin{equation}\label{4.12}
\arraycolsep=1.5pt
\begin{array}{rl}\displaystyle
r_{jl}\left \langle  \bar{L}\omega_{jl},\omega_{jl}\right\rangle
&=-\ds\sum_{k=1}^Nd_{jk}\left \langle \bar{L}D_{jk},\omega_{jl}\right\rangle-\left \langle \bar{G},\omega_{jl}\right\rangle
-\ds\sum_{s\neq l}r_{js}\left \langle  \bar{L}\omega_{js},\omega_{jl}\right\rangle\\
&\,\,\,\,\,\,\,\,+\ds\sum_{k=1}^Nc_{jk}\left \langle D_{jk},\omega_{jl}\right\rangle-\left \langle \bar{L}\Phi,\omega_{jl}\right\rangle.
\end{array}
\end{equation}

By the definition of $\bar{G}$, it is easy to verify that for $j=1,\cdots,m$,
$$
|\left \langle \bar{G},\omega_{jl}\right\rangle|
\leq C e^{-\sigma\mu}e^{-\nu\gamma|P_{m+1}-P_j|}+
\Big|\Big \langle \Big(
\begin{array}{ccccc}
\epsilon P(x)U_{P_{m+1}} \\[1mm]
\epsilon Q(x)V_{P_{m+1}}\\
\end{array}
\Big),\Big(
\begin{array}{ccccc}
\omega_{jl,1} \\[1mm]
\omega_{jl,2}\\
\end{array}
\Big)\Big\rangle\Big|,
$$
since
\begin{align*}
\ds\int_{\R^N}\big(\bar{U}^3-u_{\textbf{P}_m}^3 -U_{P_{m+1}}^3\big)\omega_{jl,1}
&\leq \ds\int_{B_{\frac{\mu}{2\gamma}}(P_j)}\left(u_{\textbf{P}_m}^2 U_{P_{m+1}}+u_{\textbf{P}_m} U_{P_{m+1}}^2\right)\zeta_j\eta_{l,1}(\gamma(x-P_j))\\
&\leq C e^{-\sigma\mu}e^{-\nu\gamma|P_{m+1}-P_j|}.
\end{align*}
Similarly, one has
$$
\big|\left \langle \bar{G},\omega_{m+1,l}\right\rangle\big|
\leq C e^{-\sigma\mu}\ds\sum_{j=1}^me^{-\nu\gamma|P_{m+1}-P_j|}+
\Big|\Big \langle \Big(
\begin{array}{ccccc}
\epsilon P(x)U_{P_{m+1}} \\[1mm]
\epsilon Q(x)V_{P_{m+1}}\\
\end{array}
\Big),\Big(
\begin{array}{ccccc}
\omega_{m+1,l,1} \\[1mm]
\omega_{m+1,l,2}\\
\end{array}
\Big)\Big\rangle\Big|.
$$

Moreover, from  \eqref{2.6} ,we have
\begin{equation}\label{4.13}
\Big|\ds\sum_{k=1}^Nc_{jk}\left \langle D_{jk},\omega_{jl}\right\rangle\Big|\leq C|c_{jk}|
\leq C e^{-\sigma\mu}+C\|\bar{G}\|_*
\leq C e^{-\sigma\mu}
\end{equation}
and
$$
|\left \langle \bar{L}\Phi,\omega_{jl}\right\rangle|
=|\left \langle \bar{L}\omega_{jl},\Phi\right\rangle|
\leq C e^{-\sigma\mu}\|\Phi\|_{H^1(B_{\frac{\mu}{2\gamma}}(P_j))}.
$$

 Using system \eqref{4.4}, we can deduce
 $$
 \left \langle  \bar{L}\omega_{jl},\omega_{js}\right\rangle=
    \delta_{ls}\lambda_s\left \langle  \eta_{l},\eta_{s}\right\rangle
    +O(e^{-\sigma\mu}).
  $$

So, from the above estimates, we can infer that
\begin{equation}\label{4.14}
\left\{%
\begin{array}{ll}
    |r_{m+1,l}|
\leq C e^{-\sigma\mu}\ds\sum_{j=1}^me^{-\nu\gamma|P_{m+1}-P_j|}+C e^{-\sigma\mu}\|\Phi\|_{H^1(B_{\frac{\mu}{2\gamma}}(P_{m+1}))}\\
\,\,\,\,\,\,\,\,\,\,\,\,\,\,\,\,\,\,\,\,\,\,
+ C\Big|\Big \langle \Big(
\begin{array}{ccccc}
\epsilon P(x)U_{P_{m+1}} \\[1mm]
\epsilon Q(x)V_{P_{m+1}}\\
\end{array}
\Big),\Big(
\begin{array}{ccccc}
\omega_{m+1,l,1} \\[1mm]
\omega_{m+1,l,2}\\
\end{array}
\Big)\Big\rangle\Big|,\\[5mm]
     |r_{i,l}|
\leq C e^{-\sigma\mu}e^{-\nu\gamma|P_{m+1}-P_j|}+C e^{-\sigma\mu}\|\Phi\|_{H^1(B_{\frac{\mu}{2\gamma}}(P_j))}\\
\,\,\,\,\,\,\,\,\,\,\,\,\,\,\,\,\,
+ C\Big|\Big \langle\Big(
\begin{array}{ccccc}
\epsilon P(x)U_{P_{m+1}} \\[1mm]
\epsilon Q(x)V_{P_{m+1}}\\
\end{array}
\Big),\Big(
\begin{array}{ccccc}
\omega_{jl,1} \\[1mm]
\omega_{jl,2}\\
\end{array}
\Big)\Big\rangle\Big|
\end{array}%
\right.
\end{equation}
and then
\begin{equation}\label{4.15}
\left\{%
\begin{array}{ll}
    |d_{k+1,j}|
\leq C e^{-\sigma\mu}\ds\sum_{j=1}^me^{-\nu\gamma|P_{m+1}-P_j|}
+C e^{-\sigma\mu}\|\Phi\|_{H^1(B_{\frac{\mu}{2\gamma}}(P_{m+1}))},\\[5mm]
     |d_{i,j}|
\leq C e^{-\sigma\mu}e^{-\nu\gamma|P_{m+1}-P_j|}
+C e^{-\sigma\mu}\|\Phi\|_{H^1(B_{\frac{\mu}{2\gamma}}(P_j))}
\end{array}%
\right.
\end{equation}
for $j=1,\cdots,m,k=1,\cdots,N,l=1,\cdots,K$.

Finally we need to estimate $\Phi$. Multiplying \eqref{4.11} by $\Phi$ and integrating over $\R^N$, we have
\begin{equation}\label{4.16}
\left \langle \bar{L}\Phi,\Phi\right\rangle
=-\left \langle \bar{G},\Phi\right\rangle
-\ds\sum_{j=1}^{m+1}\ds\sum_{k=1}^{N}d_{jk}\left \langle \bar{L}D_{jk},\Phi\right\rangle
-\ds\sum_{j=1}^{m+1}\ds\sum_{l=1}^{K}r_{jl}\left \langle  \bar{L}\omega_{jl},\Phi\right\rangle.
\end{equation}

We claim that
\begin{equation}\label{4.17}
-\left \langle \bar{L}\Phi,\Phi\right\rangle
\geq c_0\|\Phi\|_{H^1(\R^N)}^2
\end{equation}
for some constant $c_0>0$ (independent of $\textbf{P}_{m+1}$).

Indeed, since the approximate solution is exponentially decaying away from the point $P_j$,
we have
\begin{equation}\label{4.18}
\ds\int_{\R^N\setminus\cup_{j}B_{\frac{\mu}{2\gamma}}(P_j)}(\bar{L}\Phi)\Phi
\geq \frac{1}{2}\ds\int_{\R^N\setminus\cup_{j}B_{\frac{\mu}{2\gamma}}(P_j)}\left(|\nabla \Phi_1|^2+\Phi_1^2+|\nabla \Phi_2|^2+\Phi_2^2\right).
\end{equation}
So we only need to prove \eqref{4.18} in the domain $\cup_{j}B_{\frac{\mu}{2\gamma}}(P_j)$.
Here we prove it by contradiction. Assume that there exist a sequence $\mu_n\rightarrow+\infty$,
and $P_j^n$ such that as $n\rightarrow\infty$,
\begin{equation}\label{4.19}
\ds\int_{B_{\frac{\mu_n}{2\gamma}}(P_j^n)}\left(|\nabla \Phi_1^n|^2+|\Phi_1^n|^2+|\nabla \Phi_2^n|^2+|\Phi_2^n|^2\right)=1
\end{equation}
and
\begin{equation}\label{4.20}
\ds\int_{B_{\frac{\mu_n}{2\gamma}}(P_j^n)}(\bar{L}\Phi^n)\Phi^n
\rightarrow0.
\end{equation}
Then we can extract from the sequence $\Phi^n(x-P_j^n)$ a subsequence which will converge weakly in $\textbf{H}$
to $\Phi^\infty$ satisfying
\begin{align*}
&\,\,\,\,\,\,\,\,\,\,\,\,\,\ds\int_{\R^N}\left(|\nabla \Phi_{\infty,1}|^2+\Phi_{\infty,1}^2-3U^2\Phi_{\infty,1}^2-\beta\Phi_{\infty,2}\Phi_{\infty,1}\right)\\
&\,\,\,\,\,\,\,\,\,\,\,\,\,\,\,\,\,+\ds\int_{\R^N}\left(|\nabla \Phi_{\infty,2}|^2+\Phi_{\infty,2}^2-3V^2\Phi_{\infty,2}^2-\beta\Phi_{\infty,1}\Phi_{\infty,2}\right)=0
\end{align*}
and from \eqref{4.9}, we can find that
$$
\Big \langle \Phi_\infty,\Big(
\begin{array}{ccccc}
\eta_{l,1} \\[1mm]
\eta_{l,2}\\
\end{array}
\Big)\Big\rangle=\Big \langle \Phi_\infty,\Big(
\begin{array}{ccccc}
\frac{\partial U}{\partial x_k} \\[1mm]
\frac{\partial V}{\partial x_k}
\end{array}
\Big)\Big\rangle=0
$$
for $l=1,\cdots,K,k=1,\cdots,N$. So we infer that
$\Phi_\infty=0$. Hence
$$
\Phi^n\rightharpoonup0\,\,\text{weakly in}\,\,\textbf{H}.
$$
As a result, as $n\rightarrow\infty$,
$$\ds\int_{B_{\frac{\mu_n}{2\gamma}}(P_j^n)}\big[3\widetilde{U}^2(\Phi_1^n)^2
+\beta\Phi_2^n\Phi_1^n+3\widetilde{V}^2(\Phi_2^n)^2
+\beta\Phi_1^n\Phi_2^n\big]\rightarrow0$$
and then by \eqref{4.20}, one has as $n\rightarrow+\infty$,
$$
\|\Phi^n\|_{H^1\big(B_{\frac{\mu_n}{2\gamma}}(P_j^n)\big)}
\rightarrow0,
$$
which contradicts to \eqref{4.19}. Thus \eqref{4.17} holds.

It follows from \eqref{4.16} and \eqref{4.17} that
\begin{equation}\label{4.21}
\arraycolsep=1.5pt
\begin{array}{rl}\displaystyle
\|\Phi\|_{H^1(\R^N)}^2
&\leq C\Big(\ds\sum_{jk}|d_{jk}||\left \langle \bar{L}D_{jk},\Phi\right\rangle|
+\ds\sum_{jl}|r_{jl}||\left \langle \bar{L}\omega_{jl},\Phi\right\rangle|+|\left \langle \bar{G},\Phi\right\rangle|\Big)\\
&\leq C\Big(\ds\sum_{jk}|d_{jk}|\|\Phi\|_{H^1(B_{\frac{\mu}{2\gamma}}(P_j))}
+\ds\sum_{jl}|r_{jl}|\|\Phi\|_{H^1(B_{\frac{\mu}{2\gamma}}(P_j))}+\|\bar{G}\|_{L^2(\R^N)}\|\Phi\|_{H^1(\R^N)}\Big).
\end{array}
\end{equation}
Using \eqref{4.21}, \eqref{4.8}, \eqref{4.14} and \eqref{4.15}, we get that
\begin{equation}\label{4.22}
\arraycolsep=1.5pt
\begin{array}{rl}\displaystyle
\|\varphi_{m+1}\|_{H^1(\R^N)}
&\leq C\Big\{e^{-\sigma\mu}\ds\sum_{j=1}^me^{-\nu\gamma|P_j-P_{m+1}|}+\ds\int_{\R^N}\epsilon\Big(|P(x)|U_{P_{m+1}}+|Q(x)|V_{P_{m+1}}\Big)\\
&\,\,\,\,\,\,\,\,\,\,\,\,\,+\Big(\ds\int_{\R^N}\epsilon^2\Big(|P(x)|^2U_{P_{m+1}}^2+|Q(x)|^2V_{P_{m+1}}^2\Big)\Big)^{\frac{1}{2}}\\
&\,\,\,\,\,\,\,\,\,\,\,\,\,+\Big(e^{-\sigma\mu}\ds\sum_{j=1}^mw(\gamma|P_j-P_{m+1}|)\Big)^{\frac{1}{2}}\Big\}.
\end{array}
\end{equation}
Since we choose $\nu>1/2$, we have
\begin{equation}\label{4.23}
\Big(\ds\sum_{j=1}^me^{-\nu\gamma|P_j-P_{m+1}|}\Big)^2\leq C\ds\sum_{j=1}^mw(\gamma|P_j-P_{m+1}|).
\end{equation}

From \eqref{4.22} and \eqref{4.23}, we infer that
\begin{equation}\label{4.24}
\arraycolsep=1.5pt
\begin{array}{rl}\displaystyle
\|\varphi_{m+1}\|_{H^1(\R^N)}
&\leq C\Big\{\Big(e^{-\sigma\mu}\ds\sum_{j=1}^mw(\gamma|P_j-P_{m+1}|)\Big)^{\frac{1}{2}}\\
&\,\,\,\,\,\,\,\,\,\,\,\,\,+\Big(\ds\int_{\R^N}\epsilon^2(|P(x)|^2U_{P_{m+1}}^2+|Q(x)|^2V_{P_{m+1}}^2)\Big)^{\frac{1}{2}}\\[5mm]
&\,\,\,\,\,\,\,\,\,\,\,\,\,+\ds\int_{\R^N}\epsilon\big(|P(x)|U_{P_{m+1}}+|Q(x)|V_{P_{m+1}}\big)\Big\}.
\end{array}
\end{equation}
Furthermore, from the estimates \eqref{4.10} and \eqref{4.14}, and taking into consideration that $\zeta_j$
is supposed in $B_{\frac{\mu}{2\gamma}}(P_{j})$, using H\"{o}lder inequality, we can get an accurate estimate on $\varphi_{m+1}$,
\begin{equation}\label{4.25}
\arraycolsep=1.5pt
\begin{array}{rl}\displaystyle
\|\varphi_{m+1}\|_{H^1(\R^N)}
&\leq C\Big\{\Big(e^{-\sigma\mu}\ds\sum_{j=1}^mw(\gamma|P_j-P_{m+1}|)\Big)^{\frac{1}{2}}\\
&\,\,\,\,\,\,\,\,\,\,\,\,\,\,+\Big(\ds\int_{\R^N}\epsilon^2(|P(x)|^2U_{P_{m+1}}^2+|Q(x)|^2V_{P_{m+1}}^2)\Big)^{\frac{1}{2}}\\[5mm]
&\,\,\,\,\,\,\,\,\,\,\,\,\,\,+\ds\sum_{j=1}^m\Big(\ds\int_{B_{\frac{\mu}{2\gamma}}(P_j)}\epsilon^2|P(x)|^2U_{P_{m+1}}^2\Big)
^{\frac{1}{2}}+\ds\sum_{j=1}^m\Big(\ds\int_{B_{\frac{\mu}{2\gamma}}(P_j)}\epsilon^2|Q(x)|^2V_{P_{m+1}}^2\Big)^{\frac{1}{2}}\Big\}.
\end{array}
\end{equation}
This concludes the proof of Lemma \ref{lem4.2}.

\end{proof}

\section{Proof of our main result}\label{s6}

In this section, first we study a maximization problem, then we prove our main result.

 Fixing
$\textbf{P}_m\in \Omega_m$, we define a new functional
\begin{equation}\label{5.1}
\mathcal{N}(\textbf{P}_m)=J(u_{\textbf{p}_m},v_{\textbf{P}_m}):\,\Omega_m\rightarrow\R
\end{equation}
and
\begin{equation}\label{5.2}
\mathcal{R}_m:=\ds\sup_{\textbf{P}_m\in \Omega_m}{\mathcal{N}(\textbf{P}_m)}.
\end{equation}
Observe that $\mathcal{N}(\textbf{P}_m)$ is continuous in $\textbf{P}_m$. We will prove below that the maximization problem has a solution.
Let $\mathcal{N}(\bar{\textbf{P}}_m)$ be the maximum where
$\bar{\textbf{P}}_m=(\bar{P}_1,\cdots,\bar{P}_m)\in \bar{\Omega}_m$, that is
$$
\mathcal{N}(\bar{\textbf{P}}_m)=\ds\max_{\textbf{P}_m\in \Omega_m}{\mathcal{N}(\textbf{P}_m)}$$
and we denote the corresponding solution by $u_{\bar{\textbf{P}}_m}$.

First we give a lemma which will be used later.
\begin{lem}\label{lem5.1}(see Lemma 2.4, \cite{aw})
For $|P_j-P_k|\geq\frac{\mu}{\gamma}$ large, we have
$$
\ds\int_{\R^N}w^3(\gamma(x-P_j))w(\gamma(x-P_k))=(\gamma_1+o(1))w(\gamma|P_j-P_k|)
$$
for some $\sigma>0$ independent of large $\mu$ and
$$
\gamma_1=\ds\int_{\R^N}w^3(x)e^{-y_1}>0.
$$
\end{lem}

Now we prove that the maximum can be attained at finite points for each $\mathcal{R}_m$.
\begin{lem}\label{lem5.2}
 Assume that $(K_1)$, $(K_2)$ and the assumptions in Proposition \ref{prop3.3} hold. Then for all $m$:\\
 (i)\,\,There exists $\textbf{P}_m\in \Omega_m$ such that
 $$
\mathcal{R}_m={\mathcal{N}(\textbf{P}_m)};
$$
(ii)\,\,There holds$$
\mathcal{R}_{m+1}>\mathcal{R}_{m}+I(U,V),
$$
where $I(U,V)$ is the energy of $(U,V)$,
$$
I(U,V)=\frac{1}{2}\ds\int_{\R^N}\left(|\nabla U|^2+|\nabla V|^2\right)+\frac{1-\beta}{2}\ds\int_{\R^N}\left(U^2+V^2\right)
-\frac{1}{4}\ds\int_{\R^N}\left(U^4+V^4\right).
$$
\end{lem}

\begin{proof}
Here we follow the proofs in \cite{aw,cps} and we need to use the estimate in Lemma \ref{lem5.1}.
To prove this lemma, we divide the proof into the following two steps.

Step 1: We first show that $\mathcal{R}_1>I(U,V)$ and $\mathcal{R}_1$ can be attained at finite points.
Similar to the proof of Lemma \ref{lem4.2}, we have
\begin{equation}\label{5.3}
\|(\phi_P,\psi_P)\|_{H^1(\R^N)}
\leq
C\epsilon\Big(\ds\int_{\R^N}\Big(|P(x)|^2U_{P}^2+|Q(x)|^2V_{P}^2\Big)\Big)^{\frac{1}{2}}
\end{equation}
for some $C>0$ independent of $P$.

Assuming that $|P|\rightarrow+\infty$, by $U_P=V_P$, we have
\begin{align*}
J(u_P,v_P)&=\frac{1}{2}\ds\int_{\R^N}(|\nabla (U_P+\phi_P)|^2+(1+\epsilon P(x))(U_P+\phi_P)^2)\\
&\,\,\,\,\,\,+\frac{1}{2}\ds\int_{\R^N}(|\nabla (V_P+\psi_P)|^2+(1+\epsilon Q(x))(V_P+\psi_P)^2)\\
&\,\,\,\,\,\,-\frac{1}{4}\ds\int_{\R^N}((U_P+\phi_P)^4+(V_P+\psi_P)^4)-\beta\ds\int_{\R^N}(U_P+\phi_P)(V_P+\psi_P)\\
&\geq I(U_P,V_P)+\frac{\epsilon}{2}\ds\int_{\R^N}(P(x)U_P^2+Q(x)V_P^2)
-C\|((\epsilon P(x)U_P,\epsilon Q(x)V_P)\|_{L^{2}(\R^{N})}^{2}\\
&\geq I(U_P,V_P)+\frac{\epsilon}{4}\ds\int_{\R^N}(P(x)U_P^2+Q(x)V_P^2)\\
&\geq I(U_P,V_P)+\frac{\epsilon}{4}\ds\int_{B_{\frac{\mu}{2}}(P)}(P(x)U_P^2+Q(x)V_P^2)\\
&\,\,\,\,\,\,-
\frac{1}{4}\ds\sup_{B_{\frac{|P|}{4}}(0)}(U_P^\frac{3}{2}+V_P^\frac{3}{2})
\ds\int_{supp(\gamma^2(P(x)+Q(x)))^-}(|P(x)|U_P^\frac{1}{2}+|Q(x)|V_P^\frac{1}{2})\\
&\geq I(U_P,V_P)+\frac{\epsilon}{4}\ds\int_{B_{\frac{\mu}{2}}(P)}(P(x)U_P^2+Q(x)V_P^2)-O\big(\epsilon e^{-\frac{9\gamma}{8}|P|}\big).
\end{align*}

By the slow decay assumption $(K_{2}),$ we have
\begin{align*}
\frac{\epsilon}{4}\ds\int_{B_{\frac{\mu}{2}}(P)}(P(x)U_P^2+Q(x)V_P^2)-O(\epsilon e^{-\frac{9\gamma}{8}|P|})>0.
\end{align*}

So we can deduce that
\begin{equation}\label{5.4}
\mathcal{R}_1\geq J(U_P,V_P)>I(U,V).
\end{equation}
Let us prove now that $\mathcal{R}_1$ can be attained at finite points. Let $\{P_j\}$ be a sequence such that
$\ds\lim_{j\rightarrow\infty}\mathcal{N}(P_j)=\mathcal{R}_1$, and assume that $|P_j|\rightarrow\infty$ as $j\rightarrow\infty$.
Then from the system satisfied by $(U_{P_j},V_{P_j})$, we have
\begin{align*}
J(u_{P_j},v_{P_j})&=\frac{1}{2}\ds\int_{\R^N}(|\nabla (U_{P_j}+\phi_{P_j})|^2+(1+\epsilon P(x))(U_{P_j}+\phi_{P_j})^2)\\
&\,\,\,\,\,\,+\frac{1}{2}\ds\int_{\R^N}(|\nabla (V_{P_j}+\psi_{P_j})|^2+(1+\epsilon Q(x))(V_{P_j}+\psi_{P_j})^2)\\
&\,\,\,\,\,\,-\frac{1}{4}\ds\int_{\R^N}((U_{P_j}+\phi_{P_j})^4+(V_{P_j}+\psi_{P_j})^4)-\beta\ds\int_{\R^N}(U_{P_j}+\phi_{P_j})(V_{P_j}+\psi_{P_j})\\
&= I(U_{P_j},V_{P_j})+\frac{\epsilon}{2}\ds\int_{\R^N}(P(x)U_{P_j}^2+Q(x)V_{P_j}^2)
+\frac{1}{2}\ds\int_{\R^N}(|\nabla \phi_{P_j}|^2+\phi_{P_j}^2)\\
&\,\,\,\,\,\,+\frac{1}{2}\ds\int_{\R^N}(|\nabla \psi_{P_j}|^2+\psi_{P_j}^2)+\epsilon\ds\int_{\R^N}(P(x)U_{P_j}\phi_{P_j}+Q(x)V_{P_j}\psi_{P_j})\\
&\,\,\,\,\,\,+\frac{\epsilon}{2}\ds\int_{\R^N}(P(x)\phi_{P_j}^2+Q(x)\psi_{P_j}^2)
-\frac{3}{2}\ds\int_{\R^N}(U_{P_j}^2\phi_{P_j}^2+V_{P_j}^2\psi_{P_j}^2)\\
&\,\,\,\,\,\,-\ds\int_{\R^N}(U_{P_j}^2\phi_{P_j}^2+V_{P_j}^2\psi_{P_j}^2)-\frac{1}{4}\ds\int_{\R^N}(\phi_{P_j}^4+\psi_{P_j}^4)
-\beta\ds\int_{\R^N}\phi_{P_j}\psi_{P_j}\\
&\leq I(U_{P_j},V_{P_j})+\frac{\epsilon}{2}\ds\int_{\R^N}(P(x)U_{P_j}^2+Q(x)V_{P_j}^2)+O(\|(\phi_{P_j},\psi_{P_j})\|_{H^1(\R^N)}^2).\\
\end{align*}
By the decay assumptions on $P(x)$, $Q(x)$ and \eqref{5.3}, we have as $|P_j|\rightarrow\infty$,
$$\frac{\epsilon}{2}\ds\int_{\R^N}(P(x)U_{P_j}^2+Q(x)V_{P_j}^2)+O(\|(\phi_{P_j},\psi_{P_j})\|_{H^1(\R^N)}^2)\rightarrow0.$$
So it follows that
\begin{equation}\label{5.5}
\mathcal{R}_1=\ds\lim_{j\rightarrow\infty}\mathcal{N}(P_j) \leq I(U,V),
\end{equation}
which yields a contradiction to \eqref{5.4}. Then $\mathcal{R}_1$ can be attained at finite points.

Step 2: Assume that there exists $\bar{\textbf{P}}_m=(\bar{P}_1,\cdots,\bar{P}_m)\in \Omega_m$ such that
$\mathcal{R}_m=\mathcal{N}(\bar{\textbf{P}}_m)$. Next, we prove that there exists $\textbf{P}_{m+1}\in \Omega_{m+1}$
such that $\mathcal{R}_{m+1}$ can be attained. Let $(P_1^{n},\cdots,P_{m+1}^n)$ be a sequence such that
\begin{equation}\label{5.6}
\mathcal{R}_{m+1}=\ds\lim_{n\rightarrow\infty}\mathcal{N}(P_1^{n},\cdots,P_{m+1}^n).
\end{equation}
We claim that $(P_1^{n},\cdots,P_{m+1}^n)$ is bounded.
 Here we prove it by an indirect method. Without loss of generality, we assume
that $|P_{m+1}^n|\rightarrow \infty$ as $n\rightarrow\infty$. In the following, we omit the index $n$ for simplicity.
Noting that from $U_{P_{m+1}}=V_{P_{m+1}}$ and Lemma \ref{lem4.2}, we have
\begin{equation}\label{5.7}
\arraycolsep=1.5pt
\begin{array}{rl}\displaystyle
&\,\,\,\,\,\,\,\,\,\,\,\,\,J(u_{\textbf{P}_{m+1}},v_{\textbf{P}_{m+1}})
=J\Big(\Big(
\begin{array}{ccccc}
u_{\textbf{P}_{m}} \\[1mm]
v_{\textbf{P}_{m}}\\
\end{array}
\Big)+\Big(
\begin{array}{ccccc}
U_{P_{m+1}} \\[1mm]
V_{P_{m+1}}\\
\end{array}
\Big)+\Big(
\begin{array}{ccccc}
\varphi_{m+1,1} \\[1mm]
\varphi_{m+1,2}\\
\end{array}
\Big)\Big)\\
&=J\Big(\Big(
\begin{array}{ccccc}
u_{\textbf{P}_{m}} \\[1mm]
v_{\textbf{P}_{m}}\\
\end{array}
\Big)+\Big(
\begin{array}{ccccc}
U_{P_{m+1}} \\[1mm]
V_{P_{m+1}}\\
\end{array}
\Big)\Big)+\ds\sum_{j=1}^{m}\ds\sum_{k=1}^{N}c_{jk}\ds\int_{\R^N}D_{jk}\varphi_{m+1}
-\left \langle \bar{G},\varphi_{m+1}\right\rangle
\\
&\,\,\,\,\,\,\,\,
+O(\|\varphi_{m+1}\|_{H^1(\R^N)}^2)\\
&=J\Big(\Big(
\begin{array}{ccccc}
u_{\textbf{P}_{m}} \\[1mm]
v_{\textbf{P}_{m}}\\
\end{array}
\Big)+\Big(
\begin{array}{ccccc}
U_{P_{m+1}} \\[1mm]
V_{P_{m+1}}\\
\end{array}
\Big)\Big)+O\Big\{e^{-\sigma\mu}\ds\sum_{j=1}^mw(\gamma|P_{m+1}-P_j|)
\\
&\,\,\,\,\,\,\,\,+\epsilon^2\ds\int_{\R^N}(|P(x)|^2U_{m+1}^2+|Q(x)|^2V_{m+1}^2)
+\Big(\epsilon\ds\int_{\R^N}\Big(|P(x)|U_{m+1}+|Q(x)|V_{m+1}\Big)\Big)^2\Big\}.
\end{array}
\end{equation}

Next we estimate $J\Big(\Big(
\begin{array}{ccccc}
u_{\textbf{P}_{m}} \\[1mm]
v_{\textbf{P}_{m}}\\
\end{array}
\Big)+\Big(
\begin{array}{ccccc}
U_{P_{m+1}} \\[1mm]
V_{P_{m+1}}\\
\end{array}
\Big)\Big)$.
By direct computation, we can find
\begin{equation}\label{5.8}
\arraycolsep=1.5pt
\begin{array}{rl}\displaystyle
&\,\,\,\,\,\,\,\,\,\,\,\,\,J\Big(\Big(
\begin{array}{ccccc}
u_{\textbf{P}_{m}} \\[1mm]
v_{\textbf{P}_{m}}\\
\end{array}
\Big)+\Big(
\begin{array}{ccccc}
U_{P_{m+1}} \\[1mm]
V_{P_{m+1}}\\
\end{array}
\Big)\Big)\\[5mm]
&=J\Big(
\begin{array}{ccccc}
u_{\textbf{P}_{m}} \\[1mm]
u_{\textbf{P}_{m}}\\
\end{array}
\Big)+I(U_{P_{m+1}},V_{P_{m+1}})+\ds\int_{\R^N}(\nabla u_{\textbf{P}_{m}}\nabla U_{P_{m+1}}
+\nabla v_{\textbf{P}_{m}}\nabla V_{P_{m+1}})\\[5mm]
&\,\,\,\,\,\,\,\,+\ds\int_{\R^N}((1+\epsilon P(x)) u_{\textbf{P}_{m}} U_{P_{m+1}}+(1+\epsilon Q(x)) v_{\textbf{P}_{m}} V_{P_{m+1}})
-\ds\int_{\R^N} u_{\textbf{P}_{m}}^3 U_{P_{m+1}}\\[5mm]
&\,\,\,\,\,\,\,\,-\ds\int_{\R^N} v_{\textbf{P}_{m}}^3 V_{P_{m+1}}-\beta\ds\int_{\R^N} (u_{\textbf{P}_{m}} V_{P_{m+1}}+v_{\textbf{P}_{m}} U_{P_{m+1}})\\[5mm]
&\,\,\,\,\,\,\,\,-\ds\frac{1}{4}\ds\int_{\R^N} (6u_{\textbf{P}_{m}}^2 U_{P_{m+1}}^2+4u_{\textbf{P}_{m}} U_{P_{m+1}}^3)
-\frac{1}{4}\ds\int_{\R^N} (6v_{\textbf{P}_{m}}^2 V_{P_{m+1}}^2+4v_{\textbf{P}_{m}} V_{P_{m+1}}^3)\\[5mm]
&\,\,\,\,\,\,\,\,+\ds\frac{\epsilon}{2}\ds\int_{\R^N}(P(x) U_{P_{m+1}}^2+Q(x) V_{P_{m+1}}^2)\\[5mm]
&=J\Big(
\begin{array}{ccccc}
u_{\textbf{P}_{m}} \\[1mm]
u_{\textbf{P}_{m}}\\
\end{array}
\Big)+I(U_{P_{m+1}},V_{P_{m+1}})
+\ds\sum_{j=1}^{m}\ds\sum_{k=1}^{N}c_{jk}\Big \langle D_{jk},\Big(
\begin{array}{ccccc}
U_{P_{m+1}} \\[1mm]
V_{P_{m+1}}\\
\end{array}
\Big)\Big\rangle\\[5mm]
&\,\,\,\,\,\,\,\,+\ds\frac{\epsilon}{2}\ds\int_{\R^N}(P(x) U_{P_{m+1}}^2+Q(x) V_{P_{m+1}}^2)-\frac{1}{4}\ds\int_{\R^N} (6u_{\textbf{P}_{m}}^2 U_{P_{m+1}}^2+ 6v_{\textbf{P}_{m}}^2 V_{P_{m+1}}^2)\\[5mm]
&\,\,\,\,\,\,\,\,
-\ds\int_{\R^N} (u_{\textbf{P}_{m}} U_{P_{m+1}}^3+v_{\textbf{P}_{m}} V_{P_{m+1}}^3).
\end{array}
\end{equation}

By \eqref{4.13} and the definition of $D_{jk}$, we have
\begin{equation}\label{5.9}
\arraycolsep=1.5pt
\begin{array}{rl}\displaystyle
\,\,\,\,\,\ds\sum_{j=1}^{m}\ds\sum_{k=1}^{N}c_{jk}\Big \langle D_{jk},\Big(
\begin{array}{ccccc}
U_{P_{m+1}} \\[1mm]
V_{P_{m+1}}\\
\end{array}
\Big)\Big\rangle
\leq Ce^{-\sigma\mu}\ds\sum_{j=1}^mw(\gamma|P_{m+1}-P_j|).
\end{array}
\end{equation}
Similar to \eqref{4.7}, we can also get
$$
\int_{\R^N}( u_{\textbf{P}_{m}}^2 U_{P_{m+1}}^2+v_{\textbf{P}_{m}}^2 V_{P_{m+1}}^2)
\leq Ce^{-\sigma\mu}\ds\sum_{j=1}^mw(\gamma|P_{m+1}-P_j|).
$$
Moreover, by Lemma \ref{lem5.1}, we find
\begin{align*}
&\,\,\,\,\,\,\,\,\ds\int_{\R^N} (u_{\textbf{P}_{m}} U_{P_{m+1}}^3+v_{\textbf{P}_{m}} V_{P_{m+1}}^3)=2\ds\int_{\R^N}\ds\sum_{j=1}^mU_{P_{m+1}}^3U_{P_j}+\ds\int_{\R^N}(U_{P_{m+1}}^3\phi_{\textbf{P}_{m}}+V_{P_{m+1}}^3\psi_{\textbf{P}_{m}})\\
&=\ds\int_{\R^N}(U_{P_{m+1}}^3\phi_{\textbf{P}_{m}}+V_{P_{m+1}}^3\psi_{\textbf{P}_{m}})
+2\gamma^4\gamma_1\ds\sum_{j=1}^mw(\gamma|P_{m+1}-P_j|)+O(e^{-\sigma\mu})\ds\sum_{j=1}^mw(\gamma|P_{m+1}-P_j|).
\end{align*}

Therefore it follows from \eqref{5.8} that
\begin{equation}\label{5.10}
\arraycolsep=1.5pt
\begin{array}{rl}\displaystyle
&\,\,\,\,\,\,\,\,J\Big(\Big(
\begin{array}{ccccc}
u_{\textbf{P}_{m}} \\[1mm]
v_{\textbf{P}_{m}}\\
\end{array}
\Big)+\Big(
\begin{array}{ccccc}
U_{P_{m+1}} \\[1mm]
V_{P_{m+1}}\\
\end{array}
\Big)\Big)\\[5mm]
&\leq J\Big(
\begin{array}{ccccc}
u_{\textbf{P}_{m}} \\[1mm]
v_{\textbf{P}_{m}}\\
\end{array}
\Big)+I(U_{P_{m+1}},V_{P_{m+1}})
+\ds\frac{\epsilon}{2}\ds\int_{\R^N}(P(x) U_{P_{m+1}}^2+Q(x) V_{P_{m+1}}^2)\\
&\,\,\,\,\,\,\,-\ds\int_{\R^N}(U_{P_{m+1}}^3\phi_{\textbf{P}_{m}}+V_{P_{m+1}}^3\psi_{\textbf{P}_{m}})
-2\gamma^4\gamma_1\ds\sum_{i=1}^kw(\gamma|P_{m+1}-P_j|)\\
&\,\,\,\,\,\,\,+O(e^{-\sigma\mu}\ds\sum_{i=1}^kw(\gamma|P_{m+1}-P_j|)).
\end{array}
\end{equation}

Since $U_{\textbf{P}_{m}}=V_{\textbf{P}_{m}}$, by systems  \eqref{1.9} and \eqref{3.1}, we see that
\begin{align*}
&\,\,\,\,\,\,\,\,\,\ds\int_{\R^N}(U_{P_{m+1}}^3\phi_{\textbf{P}_{m}}+V_{P_{m+1}}^3\psi_{\textbf{P}_{m}})
=\ds\int_{\R^N}\Big\{\big(-\Delta+(1-\beta)\big)\Big(
\begin{array}{ccccc}
U_{P_{m+1}} \\[1mm]
V_{P_{m+1}}\\
\end{array}
\Big)\Big\}\Big(
\begin{array}{ccccc}
\phi_{\textbf{P}_{m}} \\[1mm]
\psi_{\textbf{P}_{m}}\\
\end{array}
\Big)\\
&=\ds\int_{\R^N}\Big\{(-\Delta+(1-\beta))\Big(
\begin{array}{ccccc}
\phi_{\textbf{P}_{m}} \\[1mm]
\psi_{\textbf{P}_{m}}\\
\end{array}
\Big)\Big\}\Big(
\begin{array}{ccccc}
U_{P_{m+1}} \\[1mm]
V_{P_{m+1}}\\
\end{array}
\Big)\\
&=\ds\int_{\R^N}\left(
\begin{array}{ccccc}
\ds\sum_{j=1}^{m}\ds\sum_{k=1}^{N}c_{jk}D_{jk,1}-\Big(\ds\sum_{j=1}^mU_{P_j}^3-u_{\textbf{P}_{m}}^3\Big)\\[3mm]
\ds\sum_{j=1}^{m}\ds\sum_{k=1}^{N}c_{jk}D_{jk,2}-\Big(\ds\sum_{j=1}^mV_{P_j}^3-v_{\textbf{P}_{m}}^3\Big)\\
\end{array}
\right)\Big(
\begin{array}{ccccc}
U_{P_{m+1}} \\[1mm]
V_{P_{m+1}}\\
\end{array}
\Big)\\[3mm]
&\,\,\,\,\,\,\,\,\,\,\,\,\,-\ds\int_{\R^N}\Big(
\begin{array}{ccccc}
\epsilon P(x)u_{\textbf{P}_{m}}\\[1mm]
\epsilon Q(x)v_{\textbf{P}_{m}}\\
\end{array}
\Big)\Big(
\begin{array}{ccccc}
U_{P_{m+1}} \\[1mm]
V_{P_{m+1}}\\
\end{array}
\Big).
\end{align*}

Using \eqref{3.6}, we have
\begin{align*}
&\,\,\,\,\,\,\,\Big|\ds\int_{\R^N}\epsilon P(x)u_{\textbf{P}_{m}}U_{P_{m+1}}\Big|=\Big|\ds\int_{\R^N}\epsilon P(x)\Big(\ds\sum_{j=1}^mU_{P_j}+\phi_{\textbf{P}_{m}}\Big)U_{P_{m+1}}\Big|\\
&\leq Ce^{-\sigma\mu}\ds\int_{\R^N}|P(x)|\ds\sum_{j=1}^me^{-\gamma|x-P_j|}U_{P_{m+1}}+C\|\phi_{\textbf{P}_{m}}\|_*\ds\int_{\R^N}|P(x)|\ds\sum_{j=1}^me^{-\nu\gamma|x-P_j|}U_{P_{m+1}}\\
&\leq Ce^{-\sigma\mu}\ds\sum_{j=1}^m\ds\int_{\R^N}|P(x)|e^{-\nu\gamma|x-P_j|}U_{P_{m+1}}
\end{align*}
and
$$
\Big|\ds\int_{\R^N}\Big(\ds\sum_{j=1}^mU_{P_j}^3-u_{\textbf{P}_{m}}^3\Big)U_{P_{m+1}}\Big|
\leq Ce^{-\sigma\mu}\ds\sum_{j=1}^mw(\gamma|P_{m+1}-P_j|).
$$
So from \eqref{5.9}, \eqref{5.10} and the above estimates, one has
\begin{equation}\label{5.11}
\arraycolsep=1.5pt
\begin{array}{rl}\displaystyle
&\,\,\,\,\,\,\,\,J\Big(\Big(
\begin{array}{ccccc}
u_{\textbf{P}_{m}} \\[1mm]
v_{\textbf{P}_{m}}\\
\end{array}
\Big)+\Big(
\begin{array}{ccccc}
U_{P_{m+1}} \\[1mm]
V_{P_{m+1}}\\
\end{array}
\Big)\Big)\\[5mm]
&\leq \mathcal{R}_m+I(U,V)+\ds\frac{\epsilon}{2}\ds\int_{\R^N}(P(x) U_{P_{m+1}}^2+Q(x) V_{P_{m+1}}^2)-2\gamma^4\gamma_1\ds\sum_{j=1}^mw(\gamma|P_{m+1}-P_j|)\\
&\,\,\,\,\,\,\,\,+
C\Big\{e^{-\sigma\mu}\ds\sum_{j=1}^m\ds\int_{\R^N}e^{-\nu\gamma|x-P_j|}\big(|P(x)U_{P_{m+1}}
+|Q(x)|V_{P_{m+1}}\big)
+e^{-\sigma\mu}\ds\sum_{j=1}^mw(\gamma|P_{m+1}-P_j|)\Big\}.
\end{array}
\end{equation}
As a result,
\begin{equation}\label{5.12}
\arraycolsep=1.5pt
\begin{array}{rl}\displaystyle
&J\Big(
\begin{array}{ccccc}
u_{\textbf{P}_{m+1}} \\[1mm]
v_{\textbf{P}_{m+1}}\\
\end{array}
\Big)\\&\leq \mathcal{R}_m+I(U,V)+\ds\frac{\epsilon}{2}\ds\int_{\R^N}(P(x) U_{P_{m+1}}^2+Q(x) V_{P_{m+1}}^2)-2\gamma^4\gamma_1\ds\sum_{j=1}^mw(\gamma|P_{m+1}-P_j|)\\
&\,\,\,\,\,\,\,\,+
C\Big\{e^{-\sigma\mu}\ds\sum_{j=1}^m\ds\int_{\R^N}e^{-\nu\gamma|x-P_j|}\big(|P(x)|U_{P_{m+1}}+|Q(x)|V_{P_{m+1}}\big)
+e^{-\sigma\mu}\ds\sum_{j=1}^mw(\gamma|P_{m+1}-P_j|)\\
&\,\,\,\,\,\,\,\,
+\Big(\ds\int_{\R^N}\epsilon(|P(x)|U_{P_{m+1}}+|Q(x)|V_{P_{m+1}})\Big)^2+\epsilon^2\ds\int_{\R^N}\big(|P(x)|^2U_{P_{m+1}}^2+|Q(x)|^2V_{P_{m+1}}^2\big)\Big\}.
\end{array}
\end{equation}
Since we assume that $|P_{m+1}^n|\rightarrow\infty$, we deduce that
\begin{align*}
&\,\,\,\,\,\,\,\,\epsilon\ds\int_{\R^N}(P(x) U_{P_{m+1}}^2+Q(x) V_{P_{m+1}}^2)+C\Big\{\epsilon^2\ds\int_{\R^N}(|P(x)|^2U_{P_{m+1}}^2+|Q(x)|^2V_{P_{m+1}}^2)\\
&\,\,\,\,\,\,\,\,+e^{-\sigma\mu}\ds\sum_{j=1}^m\ds\int_{\R^N}(|P(x)|e^{-\nu\gamma|x-P_j|}U_{P_{m+1}}+|Q(x)|e^{-\nu\gamma|x-P_j|}V_{P_{m+1}})\\
&\,\,\,\,\,\,\,\,+\Big(\epsilon\ds\int_{\R^N}(|P(x)|U_{P_{m+1}}+|Q(x)|V_{P_{m+1}})\Big)^2\Big\}\rightarrow0,\,\,\text{as}\,n\rightarrow\infty
\end{align*}
and
$$-2\gamma^4\gamma_1\ds\sum_{j=1}^mw(\gamma|P_{m+1}-P_j|)+O(e^{-\sigma\mu})\ds\sum_{j=1}^mw(\gamma|P_{m+1}-P_j|)<0.$$
Thus we can deduce
\begin{equation}\label{5.13}
\mathcal{R}_{m+1}\leq\mathcal{R}_m+I(U,V).
\end{equation}

On the other hand, by the assumption, $\mathcal{R}_m$ can be attained at $(\bar{P}_1,\cdots,\bar{P}_m)$.
So there exists other point $P_{m+1}$ which is far away from the $m$ points and determined later. Let us consider
the solution concentrating at the point $(\bar{P}_1,\cdots,\bar{P}_m,P_{m+1})$. We denote the solution by
$(u_{\bar{\textbf{P}}_{m},P_{m+1}},v_{\bar{\textbf{P}}_{m},P_{m+1}})$.
By similar argument as above, using the estimate \eqref{4.25} instead of \eqref{4.24}, we have
\begin{equation}\label{5.14}
\arraycolsep=1.5pt
\begin{array}{rl}\displaystyle
&\,\,\,\,\,\,\,\,J\Big(
\begin{array}{ccccc}
u_{\bar{\textbf{P}}_{m},P_{m+1}} \\[1mm]
v_{\bar{\textbf{P}}_{m},P_{m+1}}\\
\end{array}
\Big)\\[5mm]
&= J(u_{\bar{\textbf{P}}_{m}},v_{\bar{\textbf{P}}_{m}})
+I(U,V)+\ds\frac{\epsilon}{2}\ds\int_{\R^N}(P(x) U_{P_{m+1}}^2+Q(x) V_{P_{m+1}}^2)\\[3mm]
&\,\,\,\,\,\,\,\,-2\gamma^4\gamma_1\ds\sum_{j=1}^mw(\gamma|P_{m+1}-\bar{P}_{j}|)+
C\Big\{e^{-\sigma\mu}\ds\sum_{j=1}^m\ds\int_{\R^N}|P(x)|e^{-\nu\gamma|x-\bar{P}_{j}|}U_{P_{m+1}}\\[3mm]
&\,\,\,\,\,\,\,\,+e^{-\sigma\mu}\ds\sum_{j=1}^m\ds\int_{\R^N}|Q(x)|e^{-\nu\gamma|x-\bar{P}_{j}|}
V_{P_{m+1}}+e^{-\sigma\mu}\ds\sum_{j=1}^mw(\gamma|P_{m+1}-\bar{P}_{j}|)\\
&\,\,\,\,\,\,\,\,
+\Big(\epsilon\ds\sum_{j=1}^m\Big(\ds\int_{B_{\frac{\mu}{2\gamma}}(\bar{P}_{j})}|P(x)|^2U_{P_{m+1}}^2\Big)^{\frac{1}{2}}\Big)^2
+\Big(\epsilon\ds\sum_{j=1}^m\Big(\ds\int_{B_{\frac{\mu}{2\gamma}}(\bar{P}_{j})}|Q(x)|^2U_{P_{m+1}}^2\Big)^{\frac{1}{2}}\Big)^2\\[3mm]
&\,\,\,\,\,\,\,\,+\Big(\epsilon\ds\int_{\R^N}(|P(x)|U_{P_{m+1}}+|Q(x)|V_{P_{m+1}})\Big)^2\Big\}.
\end{array}
\end{equation}

By $(K_2)$, choosing $\epsilon>\bar{\epsilon}$ and $P_{m+1}$ such that $
|P_{m+1}|>> \frac{\max_{j=1}^m|\bar{P}_{j}|+\ln \epsilon}{\epsilon-\bar{\epsilon}},
$
 we can get that
\begin{align*}
&\,\,\,\,\,\,\,\,\frac{\epsilon}{2}\ds\int_{\R^N}(P(x) U_{P_{m+1}}^2+Q(x) V_{P_{m+1}}^2)
-2\gamma^4\gamma_1\ds\sum_{j=1}^mw(\gamma|P_{m+1}-\bar{P}_{j}|)\\
&\,\,\,\,\,\,\,+
C\Big\{e^{-\sigma\mu}\ds\sum_{j=1}^m\ds\int_{\R^N}(|P(x)|e^{-\nu\gamma|x-\bar{P}_{j}|}U_{P_{m+1}}+|Q(x)|e^{-\nu\gamma|x-\bar{P}_{j}|}
V_{P_{m+1}})\\
&\,\,\,\,\,\,\,\,+\Big(\epsilon\ds\int_{\R^N}(|P(x)|U_{P_{m+1}}+|Q(x)|V_{P_{m+1}})\Big)^2
+e^{-\sigma\mu}\ds\sum_{j=1}^mw(\gamma|P_{m+1}-\bar{P}_{j}|)\\
&\,\,\,\,\,\,\,\,
+\Big(\epsilon\ds\sum_{j=1}^m\Big(\ds\int_{B_{\frac{\mu}{2\gamma}}(\bar{P}_{j})}|P(x)|^2U_{P_{m+1}}^2\Big)^{\frac{1}{2}}\Big)^2
+\Big(\epsilon\ds\sum_{j=1}^m\Big(\ds\int_{B_{\frac{\mu}{2\gamma}}(\bar{P}_{j})}|Q(x)|^2U_{P_{m+1}}^2\Big)^{\frac{1}{2}}\Big)^2\Big\}\\
&\,\,\,\,\,\,\,\,
>Ce^{-\alpha\gamma|P_{m+1}|}-C\ds\sum_{j=1}^me^{-\nu\gamma|P_{m+1}-\bar{P}_{j}|}
>0.
\end{align*}
As s consequence, $$
\mathcal{R}_{m+1}\geq J\Big(
\begin{array}{ccccc}
u_{\bar{\textbf{P}}_{k},P_{m+1}} \\[1mm]
v_{\bar{\textbf{P}}_{k},P_{m+1}}\\
\end{array}
\Big)
>\mathcal{R}_m
+I(U,V),
$$
which contradicts to \eqref{5.13}. Then $\mathcal{R}_{m+1}$ can be attained at finite points in
$\Omega_{m+1}$.

Moreover, from the proof above, we can infer that
\begin{equation}\label{5.15}
\mathcal{R}_{m+1}\geq\mathcal{R}_m
+I(U,V).
\end{equation}

\end{proof}

Next we have the following proposition:

\begin{prop}\label{prop5.3}
 The maximization problem
 \begin{equation}\label{5.16}
 \ds\max_{\textbf{P}\in\bar{\Omega}_m}\mathcal{N}(\textbf{P})
 \end{equation}
 has a solution $\textbf{P}\in\Omega_m^0$, i.e., the interior of $\Omega_m$.

 \end{prop}

 \begin{proof}
We prove it by contradiction. If $\bar{\textbf{P}}_m=(\bar{P}_1,\cdots,\bar{P}_m)\in\partial \Omega_m$,
then there exists $(j,k)$ such that $|\bar{P}_j-\bar{P}_k|=\mu/\gamma$. Without loss of generality, we assume that
$(j,k)=(j,m)$. It follows from \eqref{5.12} that
\begin{equation}\label{5.17}
\arraycolsep=1.5pt
\begin{array}{rl}\displaystyle
&\,\,\,\,\,\,\,\mathcal{R}_m=J\Big(
\begin{array}{ccccc}
u_{\bar{\textbf{P}}_{m}} \\[1mm]
v_{\bar{\textbf{P}}_{m}}\\
\end{array}
\Big)\\[5mm]
&\leq \mathcal{R}_{m-1}+I(U,V)+\ds\frac{\epsilon}{2}\ds\int_{\R^N}(P(x) U_{P_{m+1}}^2+Q(x) V_{P_{m+1}}^2)-2\gamma^4\gamma_1\ds\sum_{i=1}^{m-1}w(\gamma|\bar{P}_{m}-\bar{P}_{i}|)\\
&\quad+
C\Big(e^{-\sigma\mu}\ds\sum_{i=1}^{m-1}\ds\int_{\R^N}\big(|P(x)|e^{-\nu\gamma|x-\bar{P}_{i}|}U_{\bar{P}_{m}}
+|Q(x)|e^{-\nu\gamma|x-\bar{P}_{i}|}V_{\bar{P}_{m}}\big)
+e^{-\sigma\mu}\ds\sum_{i=1}^{m-1}w(\gamma|\bar{P}_{m}-\bar{P}_{i}|)\\
&\,\,\,\,\,\,\,+\epsilon^2\ds\int_{\R^N}\big(|P(x)|^2U_{\bar{P}_{m}}^2+|Q(x)|^2V_{\bar{P}_{m}}^2\big)+\Big(\ds\int_{\R^N}\epsilon(|P(x)|U_{\bar{P}_{m}}
+|Q(x)|V_{\bar{P}_{m}})\Big)^2\Big).
\end{array}
\end{equation}

By the definition of the configuration set, we observe that given a ball of size $\mu$,
there are at most $C_N:=6^N$ number of non-overlapping balls of size $\mu$ surrounding this ball.
Using $|\bar{P}_j-\bar{P}_m|=\mu/\gamma$, we have
\begin{align*}
\ds\sum_{i=1}^{m-1}w(\gamma|\bar{P}_{m}-\bar{P}_{i}|)=w(\gamma|\bar{P}_{m}-\bar{P}_{j}|)
+\ds\sum_{i\neq j}w(\gamma|\bar{P}_{m}-\bar{P}_{i}|)\leq w(\mu)+Ce^{-\mu},
\end{align*}
since
\begin{align*}
\ds\sum_{i\neq j}w(\gamma|\bar{P}_{m}-\bar{P}_{i}|)
&\leq Ce^{-\mu}+C_Ne^{-\mu-\frac{1}{2}\mu}+\cdots+C_N^ie^{-\mu-\frac{i}{2}\mu}+\cdots\\
&\leq Ce^{-\mu}\ds\sum_{i=0}^\infty e^{i\ln C_N-\frac{1}{2}\mu}\leq Ce^{-\mu},
\end{align*}
if $C_N<e^{\frac{\mu}{2}}$ and $\mu$ large enough.

So,
\begin{align*}
\mathcal{R}_m
&\leq \mathcal{R}_{m-1}+I(U,V)+C\epsilon-2\gamma^4\gamma_1w(\mu)-2\gamma^4\gamma_1e^{-\mu}+O(e^{-\sigma\mu})w(\mu)+O(e^{-(1+\sigma)\mu})\\
&<\mathcal{R}_{m-1}+I(U,V),
\end{align*}
which yields a contradiction to Lemma \ref{lem5.2}. Then the proof is complete.

\end{proof}

Now we apply all the results obtained before to prove Theorem \ref{thm1.1}.

\begin{proof}[\textbf{Proof of Theorem \ref{thm1.1}}]
By Proposition \ref{prop3.3}, there exists $\mu_0$ such that for $\mu>\mu_0$, we have
a $C^1\times C^1$ map $(\phi_{\textbf{P}^0},\psi_{\textbf{P}^0})$ for any $\textbf{P}^0\in \Omega_m$ such that

\begin{equation}\label{6.1}
G\Big(
\begin{array}{ccccc}
U_{\textbf{P}^0}+\phi_{\textbf{P}^0} \\[1mm]
V_{\textbf{P}^0}+\psi_{\textbf{P}^0}\\
\end{array}
\Big)=\ds\sum_{s=1}^{m}\ds\sum_{l=1}^{N}c_{sl}D_{sl},\,\,\Big \langle \Big(
\begin{array}{ccccc}
\phi_{\textbf{P}^0} \\[1mm]
\psi_{\textbf{P}^0}\\
\end{array}
\Big),D_{sl}\Big\rangle=0,
\end{equation}
for some constants $\{c_{sl}\}\in\R^{m\times N}$.

From Proposition \ref{prop5.3}, there is $\textbf{P}^0\in \Omega_m^0$
that achieves the maximum of $\mathcal{N}(\textbf{P})$ given by Lemma \ref{lem5.2}.
Letting $\Big(
\begin{array}{ccccc}
u_{\textbf{P}^0} \\[1mm]
v_{\textbf{P}^0}\\
\end{array}
\Big)=\Big(
\begin{array}{ccccc}
U_{\textbf{P}^0}+\phi_{\textbf{P}^0} \\[1mm]
V_{\textbf{P}^0}+\psi_{\textbf{P}^0}\\
\end{array}
\Big)$, we have
$$
D_{P_{jk}}|_{P_j=P_j^0}\mathcal{N}(\textbf{P}^0)=0,\,\,j=1,\cdots,m,k=1,\cdots,N.
$$
Hence we have
\begin{align*}
&\,\,\,\,\,\,\,\,\,\,\,\,\,\ds\int_{\R^N}\Big(\nabla u_{\textbf{P}}\nabla\frac{\partial(U_{\textbf{P}}+\phi_{\textbf{P}})}{\partial P_{jk}}\Big{|}_{P_j=P_j^0}
+(1+\epsilon P(x))u_{\textbf{P}}\frac{\partial(U_{\textbf{P}}+\phi_{\textbf{P}})}{\partial P_{jk}}\Big{|}_{P_j=P_j^0}\Big)\\
&\,\,\,\,\,\,\,\,\,\,\,\,\,+\ds\int_{\R^N}\Big(\nabla v_{\textbf{P}}\nabla\frac{\partial(V_{\textbf{P}}+\psi_{\textbf{P}})}{\partial P_{jk}}\Big{|}_{P_j=P_j^0}
+(1+\epsilon Q(x))v_{\textbf{P}}\frac{\partial(V_{\textbf{P}}+\psi_{\textbf{P}})}{\partial P_{jk}}\Big{|}_{P_j=P_j^0}\Big)\\
&\,\,\,\,\,\,\,\,\,\,\,\,\,-\ds\int_{\R^N}\Big(u_{\textbf{P}}^3\frac{\partial(U_{\textbf{P}}+\phi_{\textbf{P}})}{\partial P_{jk}}\Big{|}_{P_j=P_j^0}
+v_{\textbf{P}}^3\frac{\partial(V_{\textbf{P}}+\psi_{\textbf{P}})}{\partial P_{jk}}\Big{|}_{P_j=P_j^0}\Big)\\
&\,\,\,\,\,\,\,\,\,\,\,\,\,-\beta\ds\int_{\R^N}\Big(v_{\textbf{P}}\frac{\partial(U_{\textbf{P}}+\phi_{\textbf{P}})}{\partial P_{jk}}\Big{|}_{P_j=P_j^0}
+u_{\textbf{P}}\frac{\partial(V_{\textbf{P}}+\psi_{\textbf{P}})}{\partial P_{jk}}\Big{|}_{P_j=P_j^0}\Big)=0,
\end{align*}
which yields that
\begin{equation}\label{6.2}
\ds\sum_{s=1}^{m}\ds\sum_{l=1}^{N}c_{sl}\ds\int_{\R^N}\Big(D_{sl,1}\frac{\partial(U_{\textbf{P}}+\phi_{\textbf{P}})}{\partial P_{jk}}\Big{|}_{P_j=P_j^0}
+D_{sl,2}\frac{\partial(V_{\textbf{P}}+\psi_{\textbf{P}})}{\partial P_{jk}}\Big{|}_{P_j=P_j^0}\Big)=0.
\end{equation}

We claim that \eqref{6.2} is a diagonally dominant system. Indeed, since
$$\ds\int_{\R^N}\big(\phi_{\textbf{P}}D_{sl,1}+\psi_{\textbf{P}}D_{sl,2}\big)|_{P_j=P_j^0}=0,$$
we have
$$
\ds\int_{\R^N}\Big(D_{sl,1}\frac{\partial\phi_{\textbf{P}}}{\partial P_{jk}}\Big{|}_{P_j=P_j^0}
+D_{sl,2}\frac{\partial\psi_{\textbf{P}}}{\partial P_{jk}}\Big{|}_{P_j=P_j^0}\Big)=0,\,\,\text{if}\,\,s\neq j.
$$
For $s=j$, we can get that
\begin{align*}
&\,\,\,\,\,\,\Big|\ds\int_{\R^N}\Big(D_{sl,1}\frac{\partial\phi_{\textbf{P}}}{\partial P_{jk}}\Big{|}_{P_j=P_j^0}
+D_{sl,2}\frac{\partial\psi_{\textbf{P}}}{\partial P_{jk}}\Big{|}_{P_j=P_j^0}\Big)\Big|\\
&=\Big|-\ds\int_{\R^N}\Big(\phi_{\textbf{P}}\frac{\partial D_{sl,1}}{\partial P_{jk}}\Big{|}_{P_j=P_j^0}
+\psi_{\textbf{P}}\frac{\partial D_{sl,2}}{\partial P_{jk}}\Big{|}_{P_j=P_j^0}\Big)\Big|\leq C\|(\phi_{\textbf{P}},\psi_{\textbf{P}})\|_*\leq Ce^{-\sigma\mu}.
\end{align*}
For $s\neq j$, we have
\begin{align*}
&\,\,\,\,\,\,\,\,\,\,\,\,\,\Big|\ds\int_{\R^N}\Big(D_{sl,1}\frac{\partial U_{\textbf{P}}}{\partial P_{jk}}
+D_{sl,2}\frac{\partial V_{\textbf{P}}}{\partial P_{jk}}\Big)\Big|\\
&\leq C\Big|\ds\int_{\R^N}\Big(\frac{\partial U(x-P_s)}{\partial x_l}\frac{\partial U(x-P_j)}{\partial x_k}
+\frac{\partial V(x-P_s)}{\partial x_l}\frac{\partial V(x-P_j)}{\partial x_k}\Big)\Big|\\
&\leq C\ds\int_{\R^N}e^{-\gamma|x-P_s|}e^{-\gamma|x-P_j|}\leq Ce^{-\frac{\gamma|P_j-P_s|}{2}}\leq Ce^{-\frac{\mu}{2}}.
\end{align*}
For $s=j$, letting $y=x-P_j$, we have
\begin{equation}\label{6.3}
\arraycolsep=1.5pt
\begin{array}{rl}\displaystyle
&\,\,\,\,\,\,\,\,\,\,\,\,\,\ds\int_{\R^N}\Big(D_{sl,1}\frac{\partial U_{\textbf{P}}}{\partial P_{jk}}
+D_{sl,2}\frac{\partial V_{\textbf{P}}}{\partial P_{jk}}\Big)\\[5mm]
&=\ds\int_{\R^N}\Big(\zeta_j\frac{\partial U(x-P_j)}{\partial x_l}\frac{\partial U(x-P_j)}{\partial P_{jk}}
+\zeta_j\frac{\partial V(x-P_j)}{\partial x_l}\frac{\partial V(x-P_j)}{\partial P_{jk}}\Big)\\[5mm]
&=-\ds\int_{B_{\frac{\mu^2}{2(\mu+1)\gamma}}(0)}\Big(\zeta_j(y+P_j)\frac{\partial U(y)}{\partial y_l}\frac{\partial U(y)}{\partial y_k}
+\zeta_j(y+P_j)\frac{\partial V(y)}{\partial y_l}\frac{\partial V(y)}{\partial y_k}\Big)\\[8mm]
&=-\delta_{lk}\ds\int_{\R^N}\Big((\frac{\partial U}{\partial y_l})^2+(\frac{\partial V}{\partial y_k})^2\Big)+O(e^{-\sigma\mu}).
\end{array}
\end{equation}

Thus, from each $(s,l)$, the off-diagonal term gives
\begin{equation}\label{6.4}
\arraycolsep=1.5pt
\begin{array}{rl}\displaystyle
&\,\,\,\,\,\,\,\ds\sum_{s\neq j}\ds\int_{\R^N}\Big(D_{sl,1}\frac{\partial(U_{\textbf{P}}+\phi_{\textbf{P}})}{\partial P_{jk}}\Big{|}_{P_j=P_j^0}
+D_{sl,2}\frac{\partial(V_{\textbf{P}}+\psi_{\textbf{P}})}{\partial P_{jk}}\Big{|}_{P_j=P_j^0}\Big)\\[5mm]
&\,\,\,\,\,\,\,+\ds\sum_{s=j,l\neq k}\ds\int_{\R^N}\Big(D_{sl,1}\frac{\partial(U_{\textbf{P}}+\phi_{\textbf{P}})}{\partial P_{jk}}\Big{|}_{P_j=P_j^0}
+D_{sl,2}\frac{\partial(V_{\textbf{P}}+\psi_{\textbf{P}})}{\partial P_{jk}}\Big{|}_{P_j=P_j^0}\Big)\\[5mm]
&\,\,\,\,\,\,\,=O(e^{-\frac{\mu}{2}})+O(e^{-\sigma\mu})=O(e^{-\sigma\mu})
\end{array}
\end{equation}
for some $\sigma>0$. So from \eqref{6.3} and \eqref{6.4}, we see that $c_{sl}=0$ for $s=1,\cdots,m,l=1,\cdots,N$. Hence
$(
u_{\textbf{P}^0},
v_{\textbf{P}^0})$ is a solution of \eqref{1.1}. By the maximum principle, it is easy to see that $u_{\textbf{P}^0}>0$
and $v_{\textbf{P}^0}>0$. This concludes the proof of Theorem \ref{thm1.1}.

\end{proof}

\textbf{Acknowledgements}
The authors are very grateful to the referee's thoughtful
reading of details of the paper and nice suggestions to
improve the result.
This work was partially supported by NSFC (No.11301204; No.11371159),
the phD specialized grant
of the
 Ministry of Education of
 China(20110144110001),
 self-determined research funds of CCNU from the colleges' basic research and operation of MOE (CCNU14A05036) and the excellent doctorial dissertation cultivation grant from Central
China Normal University (2013YBZD15).

\end{document}